\documentclass{article}

\usepackage{amssymb}
\usepackage{comment} 
\usepackage{amsmath, amssymb, amsthm}	
\usepackage{latexsym}
\usepackage{stmaryrd}
\usepackage[dvipdfmx]{graphicx}
\usepackage{bmpsize}

\makeatletter

\newcommand{\Rmnum}[1]{\expandafter\@slowromancap\romannumeral #1@}
\makeatother

\usepackage{multicol}

\newtheorem{theorem}{Theorem}[section]
\newtheorem{conjecture}[theorem]{Conjecture}
\newtheorem{definition}[theorem]{Definition}
\newtheorem{lemma}[theorem]{Lemma}
\newtheorem{proposition}[theorem]{Proposition}

\newcommand{\bigO}{{\mathcal{O}}}
\newcommand{\lb}{\left\llbracket}
\newcommand{\rb}{\right\rrbracket}

\usepackage{color}

\usepackage{authblk}
\usepackage{url}

\title{Fractal Dimension versus Process Complexity}
\author[1,5]{Joost J. Joosten}
\author[2,5]{Fernando Soler-Toscano}
\author[3,4,5]{Hector Zenil\thanks{Corresponding author: hector.zenil@algorithmicnaturelab.org}}
\affil[1]{Departamento de L\`ogica, Hist\`oria i Filosofia de la Ci\'encia, Universitat de Barcelona, Barcelona, Spain.}
\affil[2]{Departamento de Filosof\'{\i}a, L\'ogica y Filosof\'{\i}a de la Ciencia, Universidad de Sevilla, Seville, Spain.}
\affil[3]{Unit of Computational Medicine, SciLifeLab, Department of Medicine Solna, Center for Molecular Medicine, Karolinska Institute, Stockholm, Sweden.}
\affil[4]{Department of Computer Science, University of Oxford, UK.}
\affil[5]{Algorithmic Nature Group, LABORES, Paris, France.}
\date{}

\begin{document}

\maketitle

\begin{abstract}
In this paper we look at small Turing machines (TMs) that work with just two colors (alphabet symbols) and either two or three states. For any particular such machine $\tau$ and any particular input $x$ we consider what we call the \emph{space-time} diagram which is basically the collection of consecutive tape configurations of the computation $\tau(x)$. In our setting it makes sense to define a fractal dimension for a Turing machine as the limiting fractal dimension for the corresponding space-time diagrams. It turns out that there is a very strong relation between the fractal dimension of a Turing machine of the above specified type and its runtime complexity. In particular, a TM with three states and two colors runs in at most linear time, if and only if, its dimension is 2, and its dimension is 1, if and only if, it runs in super-polynomial time and it uses polynomial space. If a TM runs in time $\bigO{(x^n)}$ we have empirically verified that the corresponding dimension is $\frac{n+1}{n}$, a result that we can only partially prove. We find the results presented here remarkable because they relate two completely different complexity measures: the geometrical fractal dimension on the one side versus the time complexity of a computation on the other side. 
\\
\medskip

\end{abstract}

\newpage

\setcounter{secnumdepth}{2}
\setcounter{tocdepth}{2}



\section*{{\LARGE Part I: Theoretical setting}}
\addcontentsline{toc}{section}{Part I}

In the first part of the paper, we shall define the basic notions we work with. In particular, we shall fix on a computational model: small Turing machines with a one-way infinite tape. For these machines, we will define so-called \emph{space-time} diagrams which are a representation of the memory state throughout time. For these diagrams we shall define a notion of fractal dimension. Next, some theoretical results are proven about this dimension.

\section{Complexity Measures}

Complexity measures are designed to capture complex behavior and quantify \emph{how} complex, according to that measure, that particular behavior is. It can be expected that different complexity measures from possibly entirely different fields are related to each other in a non-trivial fashion. This paper explores the relation between two rather different but widely studied concepts and measures of complexity. On the one hand, there is a geometrical framework in which the complexity of spatio-temporal objects is measured by their fractal dimension. On the other hand, there is the standard framework of computational (resources) complexity where the complexity of algorithms is measured by the amount of time and memory they take to be executed. 

The relation we have between both frameworks is as follows. We start in the framework of computations and algorithms and for simplicity assume that they can be modeled as using discrete time steps. Now, suppose we have some computer $\tau$ that performs a certain task $\tau(x)$ on input $x$. We can assign a spatio-temporal object to the computation corresponding to $\tau(x)$ as follows. 

We look at the spatial representation $\sigma_0$ of the memory when $\tau$ starts on input $x$. Next we look at $\sigma_1$: the spatial representation of the memory after one step in the computation and so forth for $\sigma_2, \sigma_3, \ldots$. Then we `glue' these spatial objects together into one object $\Sigma(\tau, x)$ by putting each output in time next to the other: $\langle \sigma_0, \sigma_1, \sigma_2 \ldots \rangle$. Each $\sigma_i$ can be seen as a slice of $\Sigma(\tau, x)$ of the memory at one particular time $i$ in the computation. This is why we call $\Sigma(\tau, x)$ the space-time diagram of $\tau(x)$. It is of these spatio-temporal objects and in particular the limit for $x$ going to infinity that we can sometimes compute or estimate the fractal dimension $d(\tau)$.

One can set this up in such a way that $d(\tau)$ becomes a well defined quantity. Thus, we have a translation from the computational framework to the geometrical framework. Next, one can then investigate the relation between these two frameworks, and in particular, if complex algorithms (in terms of time and space complexity) get translated to complex (in the sense of fractal dimension) space-time diagrams.

It is this main question that is being investigated in this paper. The computational model that we choose is that of Turing machines. In particular we look at small one-way infinite Turing machines (TMs) with just two or three states and a binary tape alphabet. 

For these particular machines we define a notion of dimension along the lines sketched above. In exhaustive computer experiments we compute the dimensions of all machines with at most three states. Among the various relations that we uncover is that such a TM runs in at most linear time iff  the corresponding dimension is 2. Likewise, if a TM (in general) runs in super-polynomial time and uses polynomial space, we see that the corresponding dimension is 1. 

Admittedly, the way in which fractal geometry measures complexity is not entirely clear and one could even sustain the view that fractal geometry measures something entirely else. Nonetheless, dimension is clearly related to degrees of freedom and as such related to an amount of information storage.

In~\cite{zenilca} space-time diagrams of Turing machines and
one-dimensional cellular automata were investigated in the context of algorithmic information theory. Notably an
uncompressibility test on the space-time diagrams led to a classification of the behaviour of CAs and TMs thereby identifying non-trivial behaviour~\cite{zenilchaos}. The same type of space-time diagrams were also investigated in
connection to two other seminal measures of complexity~\cite{Solomonoff:1964:FormalTheoryInductiveInference, ZvonkinLevin:1970:ComplexityFiniteObjects, Bennett88logicaldepth} connected to Kolmogorov complexity, namely Solomonoff's algorithmic probability~\cite{zenilchaos, kolmo2d} and Bennett's logical depth~\cite{chaitinbb, zeniluniversal}. Interesting connections between fractal dimension and spatio-temporal parameters have also been explored in the past~\cite{krakover,vicsek,chen}, delivering a range of applications in landscape analysis and even medicine in the study of time series.

The results presented in this paper were found by computer experiments and proven in part. To the best of our knowledge it is the first time that a relation  is studied between computational complexity and fractal geometry, of a nature as presented here.\\

{\bf Outline:} The current paper naturally falls apart into three parts. In the first part (Sections \ref{section:SmallTMdataBase}---\ref{section:SpaceTimeTheorem}) we define the ideas and concepts and prove various theoretical results. In the second part, Sections \ref{section:TheExperiment}---\ref{section:MostSalientFindings}, we describe our experiment and its results to investigate those cases where non of our theoretical results would apply. Finally in the third part, we present a literature study where we mention various results that link fractal dimension to other complexity notions.

More in detail:
In Section \ref{section:SmallTMdataBase} we describe the kind of TMs we shall work with. This paper can be seen as part of a larger project where the authors mine and study the space of small TMs. As such, various previous results and data could be re-used in this paper and in Section \ref{section:SmallTMdataBase} we give an adequate description of these used data and results.

In Section \ref{section:FractalDimension} we revisit the box-counting dimension and define a suitable similar notion of fractal dimension $d(\tau)$ for TMs $\tau$. We prove that $d(\tau) =2$ in case $\tau$ runs in time at most linear in the size of the input. Next, in Section \ref{section:SpaceTimeTheorem} we prove an upper and a lower bound for the dimension of Turing machines. The Upper Bound Conjecture is formulated to the effect that the proven upper bound is actually always attained. For special cases this can be proved. Moreover, under some additional assumptions this can also be proven in general. In our experiment we test if in our test-space the sufficient additional assumptions were also necessary ones and they turn out to be so.

Section \ref{section:TheExperiment} describes how we performed the experiment, what difficulties we encountered, how they were overcome, and also some preliminary findings are given. The main findings are presented in Section \ref{section:MostSalientFindings}.

We conclude the paper with Section \ref{section:LiteratureSurvey} where we present various results from the literature that link different notions of complexity to put our results within this panorama.

\section{The space of small Turing machines}\label{section:SmallTMdataBase}

As mentioned before, this paper forms part of a larger project where the authors exhaustively mine and investigate a set of small Turing machines. In this section, we will briefly describe the raw data that was used for the experiments in this paper and refer for details to the relevant sources.

\subsection{The model}

A TM can be conceived both as a computational device and as a dynamical system. In our studies a TM is represented by a \emph{head} moving over a \emph{tape} consisting of discrete \emph{tape-cells} where the tape extends infinitely in one direction. In our pictures and diagrams we will mostly depict the tape as extending infinitely to the left. Each tape cell can contain a symbol from an \emph{alphabet}. Instead of symbols we speak of \emph{colors} and in the current paper we shall work with just two colors: black and white.

The head of a TM can be in various \emph{states} as it moves over the cells of the tape. We shall refer to the collection of TMs that use $n$ states and $k$ symbols/colors as the $(n,k)$-space of TMs. We shall always enumerate the states from $1$ to $n$ and the colors from $0$ to $k-1$. In this paper we work with just two symbols so that we represent a cell containing a 0 with a white cell, and a cell containing a 1 with a black cell.

A computation of a TM proceeds in discrete time-steps. The tape content at the start of the computation is called the \emph{input}. By definition, our TMs will always start with the head at the position of the first tape cell, that is, the tape cell next to the edge of the tape; In our pictures this is normally the right-most tape. Moreover, by definition, our TMs will always commence their computation in the default start state $1$.

A TM $\tau$ in $(n,k)$ space is completely specified by its \emph{transition table}. 
This table tells what \emph{action} the head should perform when it is in State $1\leq j \leq n$ at some tape cell $c$ and reads there some symbol $0\leq i < k$. Such an action in turn consists of three aspects: changing to some some state (possibly the same one); the head moving either one cell left or one cell right, but never staying still; writing some symbol at $c$ (possibly the same symbol as was there before). Consequently, each $(n,k)$-space consists of $(2\cdot n\cdot k)^{n\cdot k}$ many different TMs. We number these machines according to Wolfram's enumeration scheme (\cite{wolfram02}, \cite{JoostDemo}) which is similar to the lexicographical enumeration.

Clearly, each TM in $(n,k)$ space is also present in $(m,k)$ space for $m\geq n$, by just not using the extra states since they are `inaccessible' from State 1. Many rules in a $(n,k)$ space are trivially equivalent in the computational sense up to a simple transformation of the underlying geometry, for example, by relabeling states by reflection or complementation, hence for all purposes identical. In the literature machines that have equivalents are sometimes called \emph{amphichiral}, we will sometimes refer to them as machine \emph{twins}.

We say that a TM \emph{halts} when the head ``falls off" the tape on the right-hand side; in other words, when the head is at the right-most position and receives an instruction to move right. The tape configuration upon termination of a computation is called the \emph{output}.

We shall refer to the input consisting of the first $m$ tape cells by black on an otherwise white tape as the input $m$. (This is in slight discrepancy with the convention in \cite{JoostenSZ11}.) In this context, a \emph{function} is a map sending an input $m$ to some output tape configuration. We call the function where the output is always identical to the input the \emph{tape identity} function.

By Rice's Theorem it is in principle undecidable if two TMs compute
the same function. Nonetheless, for spaces $(n,2)$ with $n$ small, no
universal computation is yet present (\cite{NWoods,Margenstern}). In \cite{JoostenSZ11} the authors completely classify the TMs in (3,2) space among the functions they compute, taking pragmatic approaches that possibly produce small errors to deal with undecidability and unfeasibility issues.

\subsection{Space-time diagrams}

As mentioned in the introduction, in this paper a central role is played by so-called \emph{space-time diagrams}. A space-time diagram for some computation is nothing more but the joint collection of consecutive memory configurations. We have included a picture of space-time diagrams for a particular TM for inputs 1 to 14  in Figure \ref{figure:TM346SpaceTimeDiagrams}.

\begin{figure}[htb!]
  \centering
  \includegraphics[height=7cm]{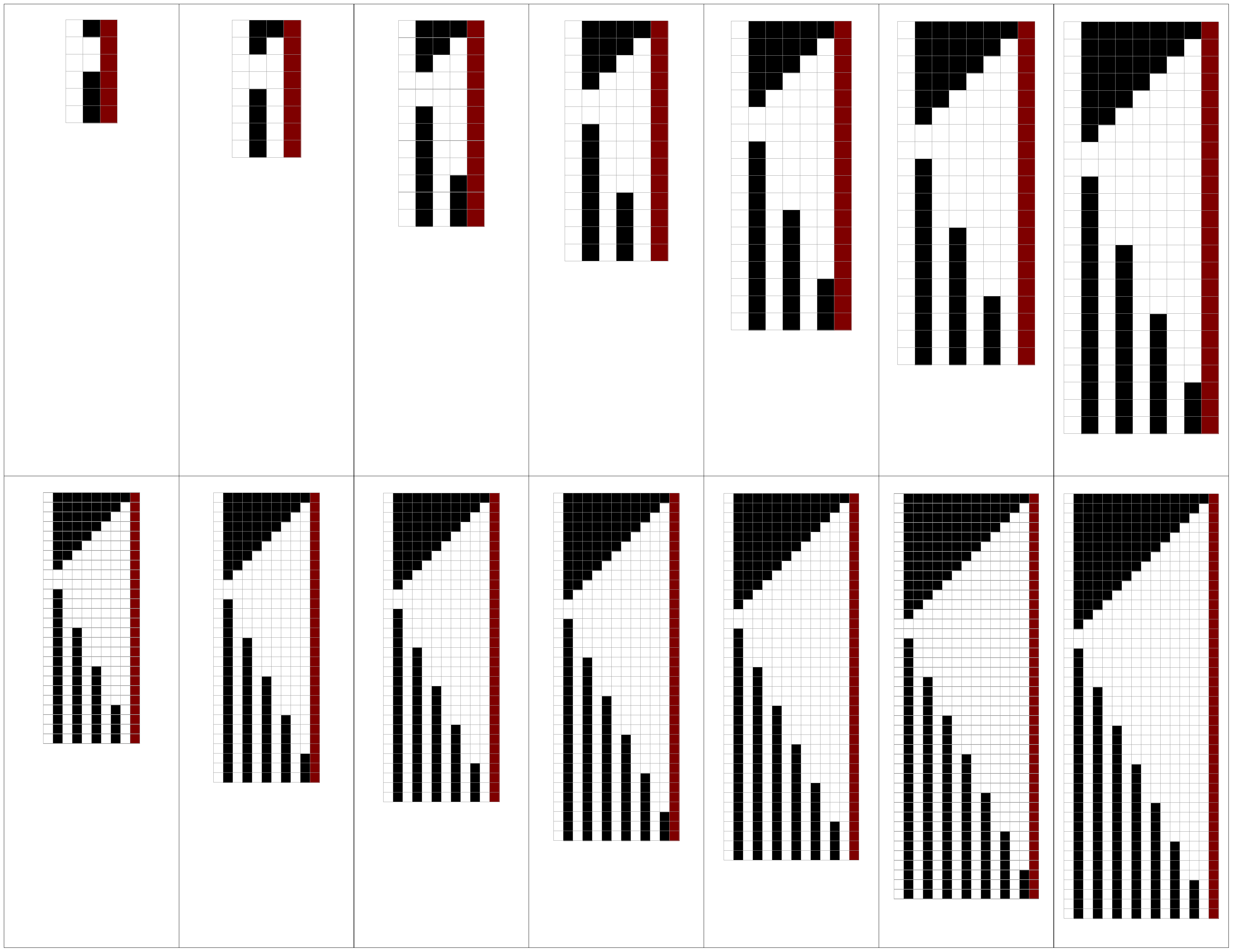}
  \caption{The figure shows a sequence of space-time diagrams corresponding to the TM in (2,2) space with number 346 (according to Wolfram's enumeration scheme \cite{wolfram02, JoostDemo} for (2,2) space) on inputs 1 up to 14.}
  \label{figure:TM346SpaceTimeDiagrams}
\end{figure}

Since these space time diagrams are such a central notion to this paper, let us briefly comment on Figure \ref{figure:TM346SpaceTimeDiagrams}. The top-row of each of these fourteen diagrams always represents the input tape configuration of the TM. We have chosen to depict the space-time diagrams of our TM on inputs 1 to 14. The rightmost cell in the diagram is actually not representing a tape cell. Rather it represents the end of the tape so that we depict it with a different color/grey-tone.

Remember that the computation starts with the head of the TM in State 1 in the rightmost cell. Each lower row represents the tape configuration of a next step in the computation. So, there can at most be one cell of different color between two adjacent rows in a space-time diagram. We see that this particular (2,2) TM with number 346 first moves over the tape input erasing it. Then it gradually moves back to the edge of the tape writing alternatingly black and white cells to eventually fall of the tape, whence it terminates.

Clearly these space-time diagrams define spatio-temporal objects by focussing on the black cells. It is of these spatio-temporal objects that we wish to measure the geometrical complexity. Subsequently, we wish to see if there is a relation between this geometrical complexity and the computational complexity (space or time usage) of the TM in question.

In the next section we shall see how to assign a measure of geometrical complexity to these space-time diagrams and call this measure the \emph{dimension} of the TM. Various relations between computational complexity of a TM on the one hand and its dimension on the other hand can be proven. Other relations shall be investigated via experiments.

\subsection{On our coding convention}

Note that for this paper it is entirely irrelevant how to numerically interpret the output tape configuration whence we shall refrain from giving such an interpretation. However, it has been a restrictive choice to represent our input in a unary way. That is to say, the notion of a function in our context only looks at a very restricted class of possible inputs: blocks of $n$ consecutive black cells for $n>0$. The main reason why we do this is that if we do not do this, our functions all behave in a very awkward and highly undesirable way. In \cite{JoostenSZ11} this undesirable behavior is explained in the so-called \emph{Strips Theorem}.

Basically, the Strips Theorem boils down to the following. Let us consider a TM $\tau$ in $(n,2)$ space on input $x$ and suppose $\tau (x)$ is a terminating computation. If we number the cells on the tape by their distance to the edge, let $i$ be largest cell-number that is visited in the computation of $\tau(x)$. Clearly, any tape input that is equal to $x$ on the first $i$ cells but possibly different on the cells beyond $i$, will perform exactly the same computation and in this sense in is input-independent.

We have chosen our input-output convention in such a way to prevent the Strips-Theorem. There are two undesirable side effects of our coding.
Firstly, it is clear that any TM that runs in less than linear time, actually runs in constant time. Secondly, the thus defined functions are very fast growing if we were to represent the output in binary. In particular, the tape identity represents an exponentially fast growing numerical function in this way.

A positive feature of our input convention is that the amount of symmetry present in the input coding facilitates various types of analysis and in particular automated function-completion seems to run more smoothly. {\bf Warning}: A clear drawback of our convention is that one tends to think of the $n$th input as the number $n$. In the context of this paper, this would not be good practice since it would, for example yield a linear primality test, factorization algorithm, etc.

We shall here describe an alternative way of representing the input $a$ and denote the representation by $\rho(a)$. This representation $\rho(a)$ shall be such that it avoids the Strips Theorem yet does not intrinsically entail exponential growth of the tape identity and similar functions in case we would interpret our output configuration in binary. Although we do not use, nor explore the alternative input coding we find it worth mentioning here and hope that future investigations can take up the new coding.

In order to represent the input $a$ according to $\rho$, we first write the input $a$ in binary as $\sum_{n=0}^{\infty}a_n2^n$ with all but finitely many $a_n = 0$. Let us denote the cells on the tape by $c_0, c_1,c_2, \ldots$. Here $c_0$ is the cell at the edge, $c_1$ the cell immediately next to it, etc. For $a\neq 0$, let $k = \lfloor \log_2 (a)\rfloor +1$ and $k=1$ otherwise. That is, $k+1$ is the number of digits in the binary expansion of $a$.

For each $0\leq i \leq k$ we shall represent $a_i$ in $c_{2\cdot i}$ in the canonical way: we set $c_{2\cdot i}$ to be one/black whenever $a_i =1$ and we set $c_{2\cdot i}$ to be zero/white otherwise. Moreover, we set all odd-labeled cells $c_{2\cdot i +1}$ to be zero with the sole exception at cell $2k+1$ that we define to be one/black.

It is clear that $\rho$ avoids the Strips Theorem. Moreover, if we were to interpret the output as binary, the tape identity defines a function whose growth rate is only in the order of $x \mapsto x^2$.

\section{Fractal dimensions}\label{section:FractalDimension}

In this section we shall briefly recall the definition of and ideas behind the box-counting dimension which is a particular fractal dimension having the better computational properties whence better suited for applications. 
In Section \ref{section:LiteratureSurvey} we relate the box-dimension to various other notions of fractal dimension and in particular to the well-known Hausdorff dimension.

After revisiting the notion of box-counting dimension, we see how to apply these ideto Turing machines and their space-time diagrams.

\subsection{Box dimension}
We shall use the notion of Box dimension. This notion of fractal dimension can be seen as a simplification of the well-known Hausdorff dimension (see \cite{hausdorff} and our survey section, Section \ref{section:LiteratureSurvey}). The Hausdorff dimension is mathematically speaking more robust than the Box dimension. However, the Box-dimension is easier to compute and is known to coincide with the Hausdorff dimension in various situations.

Let us briefly recall the definition of the Box dimension and the main ideas behind it. The intuition is as follows. Suppose we have a mathematical object $S$ of bounded size whose ``volume" $V(S)$ we wish to estimate. For example, let us work with a space $\mathbb R^n$ that has dimension $n$ large enough to embed our object $S$. The idea now is to cover the object $S$ by boxes in $\mathbb R^n$ and estimate the ``volume'' $V(S)$ of $S$ as function of the total number of boxes $N(S)$ needed to cover $S$. Clearly, the number of boxes $N(S)$ needed to cover $S$ depends on the size of the boxes used. Therefore, in the analysis we will take along the parameter $r$ which denotes the length of the edge of the boxes used, and we will write the number of boxes needed to cover $S$ as $N(S,r)$.

If $S$ is a line, which is a one-dimensional object, the corresponding notion of ``volume'' $V(S)$ is just the length of the line segment. To estimate the length $V(S)$ we clearly have 
\[
V(S) = \lim_{r\downarrow 0} \ \ \ r N(S,r) 
\]
if this is well defined.

If $S$ is a plane, or more in general a two-dimensional manifold, the corresponding notion of ``volume'' $V(S)$ is just the surface of the plane/manifold segment for which we have
\[
V(S) = \lim_{r\downarrow 0} \ \ \ r^2 N(S,r) 
\]
if well defined.

Likewise, for a three-dimensional object, to estimate its volume we would have 
\[
V(S) = \lim_{r\downarrow 0} \ \ \ r^3 N(S,r) 
\]
if well defined, and in general for an $d$-dimensional object we would obtain
\begin{equation}\label{equation:VolumeDimensionAndNumberOfBoxes}
V(S) = \lim_{r\downarrow 0} \ \ \ r^d N(S,r) .
\end{equation}

The idea behind the definition of the Box dimension is to take \eqref{equation:VolumeDimensionAndNumberOfBoxes} as a \emph{defining} equation of dimension if this makes sense mathematically speaking. Thus, solving for $d$ in \eqref{equation:VolumeDimensionAndNumberOfBoxes} we obtain:
\[
\begin{array}{rll}
\ &V(S)  =  \lim_{r\downarrow 0} \ \ \ r^d N(S,r) & \Longrightarrow \\
\ &\log \big( V(S) \big)  =  \lim_{r\downarrow 0} \ \ \ \Big(d \log(r) + \log(N(S,r))\Big) & \Longrightarrow \\
\ &\lim_{r\downarrow 0} \ \ \ \Big(d \log(r) + \log(N(S,r)) -\log \big( V(S) \big)\Big) =0 & \Longrightarrow \\
\ &\lim_{r\downarrow 0} \ \ \ \Big(d + \frac{\log(N(S,r))}{\log(r)} -\frac{\log \big( V(S) \big)}{\log(r)}\Big) =0 & \Longrightarrow \\
d&= \lim_{r\downarrow 0} \ \ \ \Big(-\frac{\log(N(S,r))}{\log(r)} +\frac{\log \big( V(S) \big)}{\log(r)}\Big) &  \\
\ &= \lim_{r\downarrow 0} \ \ \ \Big(\frac{\log(N(S,r))}{\log(\frac{1}{r})} +\frac{\log \big( V(S) \big)}{\log(r)}\Big) &  \\
\ &= \lim_{r\downarrow 0} \ \ \ \frac{\log(N(S,r))}{\log(\frac{1}{r})}. &  \\
\end{array}
\]
The last equality is justified in case $S$ is bounded as by assumption and $V(S)$ is finite so that 
$\lim_{r\downarrow 0} \  \frac{V(S)}{\log(r)} = 0$.

The above reflections form the main ideas behind the definition of Box dimension that we shall use in this paper.

\begin{definition}[Box dimension]\label{definition:BoxDimension}
Let $S$ be some spatio-temporal object that can be embedded in some $\mathbb{R}^n$, let $N(S,r)$ denote the minimal number of boxes of size $r$ needed to fully cover $S$. The Box dimension of $S$ is denoted by $\delta (S)$ and is defined by
\[
\delta (S) \ \ := \ \  \lim_{r\downarrow 0} \ \ \ \frac{\log(N(S,r))}{\log(\frac{1}{r})} 
\] 
in case this limit is well defined. In all other cases we shall say that $\delta (S)$ is undefined.
\end{definition}

\subsection{Box dimension for Space-Time diagrams}

Let us see how we can adapt the notion of Box dimension to our space-time diagrams. The spatio-temporal figure $S$ that we wish to measure will be defined by the black cells in the space-time diagram. Clearly, for each particular input on which the TM halts the corresponding space-time diagram is finite and has dimension $d(S) = 2$ : each black cell defines a piece of surface. It gets interesting when we consider limiting behavior of the TM on larger and larger inputs.

\subsubsection{A first attempt}
Let $\tau$ be some TM and let $S(\tau ,x)$ denote the space-time diagram corresponding to TM $\tau$ on input $x$ if this is well-defined, that is, if $\tau$ eventually halts on input $x$ which we shall denote by $\tau (x)\downarrow$. The question is, what is a sensible way to define the dimension $d(\tau)$ of our TM $\tau$?  It does not make much sense to define $d(\tau) = \lim_{x\to \infty} \delta (S(\tau,x))$ for a couple of reasons.

Firstly, $\tau$ might diverge on various inputs. We can easily bypass that by tacitly understanding $\lim_{x\to \infty}$ as $x$ getting larger and larger among those $x$ for which $\tau(x)\downarrow$ demanding that there are infinitely many such $x$. In case there are just finitely many $x$ on which $\tau$ converges, we could say that $d(\tau)$ is undefined.

The second objection is more serious. As for each $x$ with $\tau(x)\downarrow$ we have that $\delta(S(\tau,x))=2$ we see that all limits converge to the value 2 if they are well defined. This of course is highly undesirable. We can overcome this objection by scaling the length of each $S(\tau, x)$ to some figure ${\sf scale} (S(\tau,x))$ whose length has unit size. Thus, the black areas in ${\sf scale} (S(\tau,x))$ become more and more fine-grained so that it seems to make sense to define $d(\tau) = \delta \Big( \lim_{x\to \infty} {\sf scale} (S(\tau,x)) \Big)$.

\subsubsection{A second attempt and formal definition}

The new candidate $d(\tau) = \delta \Big( \lim_{x\to \infty} {\sf scale} (S(\tau,x)) \Big)$ has many good properties. However, for this new candidate  we again see two main objections.

The first objection is that $\lim_{x\to \infty} {\sf scale} (S(\tau,x))$ need not exist at all and stronger still, is likely not to be well defined in most cases. We could try to remedy this by working with with subsequences for which the limit is defined but it all seems very hairy.

The second objection is that this new definition seems hard to numerically approximate at first glance. We shall see how to overcome the second objection which will yield us automatically a solution to the first objection.

As we mentioned before, we cannot first approximate $\lim_{x\to \infty} {\sf scale} (S(\tau,x))$ and then compute the corresponding $\delta$ as this would always yield the answer 2. However, what we \emph{can} do is simultaneously approximate both $\delta$ and $\lim_{x\to \infty} {\sf scale} (S(\tau,x))$. There is a lot of choice in how we approximate and in how fast we approximate $\delta$ and how fast we approximate $\lim_{x\to \infty} {\sf scale} (S(\tau,x))$ relatively to the approximation of $\delta$.

There seems to be a canonical choice though. The approximation of the dimension $\delta$ is dependent on the size $r$ of the boxes. It seems very natural to take the size of our boxes to be exactly the size of one tape-cell. The size of one tape-cell is naturally determined by ${\sf scale} (S(\tau,x))$. Let us determine $r$ as dictated by ${\sf scale} (S(\tau,x))$. In order to facilitate our discussion we first fix some notation.

\begin{definition}
For $\tau$ a TM and $x$ an input so that $\tau(x) {\downarrow}$, we denote by $t(\tau,x)$ the amount of time-steps $\tau$ needed on input to terminate. Likewise, $s(\tau,x)$ denotes the amount of space-cells used by the computation of $\tau$ on input $x$. Thus, $s(\tau,x)$ measures the distance between the edge of the tape and the furthest tape cell visited by the head during the computation.

We shall sometimes write $t_\tau(x)$ or even just $t(x)$ if the context allows us to, and similarly for the space-usage function $s(\tau, x)$.
\end{definition}

By the nature of our input-output protocol, there exist no TMs whose runtime is sub-linear but not constant.  Let us first concentrate on the TMs that run in at least linear time and deal with the constant time TMs later. If a TM halts in non-constant time, the least it should do is read all the input, do some calculations and then go back to the beginning of the tape. Thus, clearly $t(\tau,x) > s(\tau, x)$, whence the scaling of the figure $S(\tau, x)$ is best done by resizing the runtime to get length 1. Consequently, the size of $r$ scales to $r= \frac{1}{t(\tau,x)}$.

Recall that $N(\Sigma,r)$ denotes the minimal number of boxes of size $r$ needed to cover the spatio-temporal object $\Sigma$. Now that we have determined the size of $r$, we can write $N(\tau, x)$ instead of $N(S(\tau, x),\frac{1}{t(\tau,x)})$ and it is clear that $N(\tau, x)$ is just the number of black cells in the space time diagram of $\tau$ on input $x$. Thus, the second attempt of defining $d(\tau)$ then translates to

\begin{equation}\label{equation:dimensionTMSecondAttempt}
d(\tau) := \lim_{x\to \infty} \frac{\log \big( N(\tau, x) \big)}{\log \big( t(\tau, x) \big)}.
\end{equation}

In this definition we could address the issue of undefinedness by replacing $\lim$ by $\limsup$ or $\liminf$. Notwithstanding the theoretical correctness of this move, it seems hardly possible to sensibly compute the $\limsup$ or $\liminf$ in the general setting.

In the current paper however, we have only considered TMs with either two or three states and just two colors. It turned out that in this setting we could determine both $\liminf$ and $\limsup$. In all cases that we witnessed where the limit \eqref{equation:dimensionTMSecondAttempt} outright was not well-defined, we were able to identify different subsequences where the limit \eqref{equation:dimensionTMSecondAttempt} did converge so that we could choose to either go with the $\limsup$ or with the $\liminf$. It turns out that in general, the lower the dimension the more interesting the corresponding TM so that we decided to work with $\liminf$.

For TMs with constant runtime, we know that only a constant number of cells will be visited and possibly changed color. For these TMs the figure $S(\tau, x)$ can only be sensibly scaled by using the input size. By doing so, we see that in the limit we just get a black line whose dimension is clearly equal to one. However, as we consider constant runtime TMs as a degenerate case so to say, we shall for convenience define the dimension of such machines to be equal to 2. We do so in order to have them more like linear time TMs (see Lemma \ref{theorem:LinearTimeTMsHaveDimension2}). All these considerations and reflections lead us to the following definition.

\begin{definition}[Box dimension of a Turing machine]
Let $\tau$ be a TM that converges on infinitely many input values $x$. In case $\tau(x)\downarrow$, let $N(\tau, x)$ denote the number of black cells in the space-time diagram of $\tau$ on input $x$ and let $t(\tau, x)$ denote the number of steps needed for $\tau$ to halt on $x$.

We will define the Box dimension of a TM $\tau$ and denote it by $d(\tau)$. In case $t(\tau,x)$ is constant from some $x$ onwards, we define $d(\tau) \ := \ 2$. Otherwise, we define
\[
d(\tau) := \liminf_{x\to \infty} \frac{\log \big( N(\tau, x) \big)}{\log \big( t(\tau, x) \big)} .
\]
\end{definition}

Note that our definition of dimension can readily be generalized to non-terminating computations. Also, restricting to computational models with discrete time steps is not strictly necessary.

\subsection{Linear time Turing machines}

For certain TMs $\tau$, we can actually compute their Box dimension. Let us reconsider TM 346 again whose space-time diagrams were displayed in Figure~\ref{figure:TM346SpaceTimeDiagrams}. Due to the extreme regularity in the space-time diagrams we can see that TM 346 runs in linear time. That is to say, linear in the length of the representation of the input.

Thus, after scaling each space-time diagram so that the vertical time-axis is rescaled to 1, we will always have a little surface in the shape of a black triangle in the scaled space-time diagram. The Box dimension of a triangle is of course 2. We may conclude that $d(2,2\text{-TM } 346)=2$. Of course the only important feature used here is the linear-time performance of 2,2-TM 346. We can summarize this observation in a lemma.

\begin{lemma}\label{theorem:LinearTimeTMsHaveDimension2}
Let $\tau$ be a TM that runs at most linear time. Then, $d(\tau) = 2$.
\end{lemma}

\begin{proof}
We fix some TM $\tau$ that runs at most linear time. Our input/output convention is such that $\tau$ is either constant time or $\tau$ is linear time. In case $\tau$ runs in constant time, we have that $d(\tau)=2$ by definition.

Let us consider the case that $\tau$ runs in linear time. It must be the case that the head goes all the way to the end of the tape input, if not, $\tau$ would run in constant time from some input $y$ onwards. Input $x$ is represented by $x+1$ consecutive black cells. In the worst case (the fewest amount of black cells), at all the first steps the input is erased and replaced by a white cell as is the case in Figure \ref{figure:TM346SpaceTimeDiagrams}.

However, $\tau$ runs in linear time, say, for any $x$, the machine $\tau$ runs at most $a\cdot(x+1)$ many steps with $2 \leq a < \infty$. After scaling, the input will get size $\frac{x+1}{a\cdot (x+1)} = \frac{1}{a}$. Thus in the worst case, the upper triangle has size $\frac{1}{2a^2}$ which is independent on $x$ whence non-vanishing. Clearly, the Box dimension of a triangle of whatever size equals 2. Thus $d(\tau) = 2$ as was what we wanted to see. 
\end{proof}

\section{The Space-Time Theorem and applications}\label{section:SpaceTimeTheorem}

Above we saw that for linear time TMs we can actually compute the corresponding dimension. However, for non-linear TMs we can only prove an upper bound on the Box dimension.

\subsection{The Space-time Theorem: an upper bound}

\begin{theorem}[Space-time Theorem]\label{theorem:SpaceTimeTheorem}
Let us, for a given TM $\tau$, denote by  $s(x)$ the amount of cells visited by $\tau$ on input $x$, and let $t(x)$ denote the amount of computation steps it took $\tau$ to terminate on input $x$.
\[
\mbox{If $\liminf_{x \to  \infty} \  \frac{\log(s(x))}{\log(t(x))} = n$ then $d(\tau) \leq  1 +n$. }
\]
\end{theorem}

\begin{proof}
The box dimension is maximal in case all cells under consideration are black. This number is bounded above by $s \cdot  t$.
Plugging this in the definition of $d(\tau)$ gives us our result:
\[
\begin{array}{lll}
d(\tau) & = & \liminf_{x\to \infty} \frac{\log \big( N(\tau, x) \big)}{\log \big( t(\tau, x) \big)}\\
\ &\leq & \liminf_{x\to \infty} \frac{\log \big( s(x) \cdot t(x) \big)}{\log \big( t(x) \big)}\\
\ &\leq & \liminf_{x\to \infty} \frac{\log \big( s(x) \big) + \log \big( t(x) \big)}{\log \big( t(x) \big)}\\
\ &\leq &  1+ \liminf_{x\to \infty} \frac{\log \big(s(x) \big)}{\log \big( t(x) \big)}\\
\ &\leq & 1+n.\\
\end{array}
\]
\end{proof}

As we shall see, in all cases the upper bound given by the Space-time Theorem is actually attained in our experiment. It is unknown however, if it holds in general.

\subsection{A lower bound}

We first observe that for any Turing machine $\tau$ we have that $d(\tau)\geq 1$. The main idea is that if there are too many white cells, then the Turing machine would enter in a loop and either finish straight away or never finish\footnote{We would like to thank an anonymous referee for suggesting this simple argument to us.}.

\begin{theorem}\label{theorem:lowerBound}
The dimension $d(\tau)\geq 1$ for any Turing machine $\tau$.
\end{theorem}

\begin{proof}
Let $\sigma$ be the number of states of some fixed TM $\tau$. It is clear that if we have a sequence of $\sigma$ consecutive steps in a computation of $\tau$ where the tape is entirely white, then $\tau$ will enter in a loop. Thus for $\tau$, in order to exhibit non-trivial behavior, we should have --modulo an additive constant-- that
\begin{equation}\label{equation:simpleLowerBound}
\frac{N(\tau, x )}{t(\tau,x)} > \frac{1}{\sigma}.
\end{equation} 
For linear or constant time TMs $\tau'$ we had already observed in Lemma \ref{theorem:LinearTimeTMsHaveDimension2} that $d(\tau') \geq 1$ so we may assume that $\tau$ has super-linear runtime asymptotic behavior. But then, from \eqref{equation:simpleLowerBound} it follows that in the limit we have $\frac{\log(N(\tau, x ))}{\log(t(\tau,x))} \geq 1$ as wto be shown.
\end{proof}

The method in proving the lower bound seems very crude: no blocks of $\sigma$ consecutive entirely white tapes may occur. It seems that more ind-depth analysis could yield sharper lower bounds.

\subsection{The Asymptotic Conjectures}

In Theorem \ref{theorem:lowerBound} we proved $d(\tau)\geq 1$. However, we conjecture that something stronger actually holds.

\begin{conjecture}[Space-time ratio conjecture]
For each TM $\tau$ which runs in more than linear time, we have that $\displaystyle \lim_{x\to \infty} \frac{s(\tau,x)}{t(\tau,x)} =0$.
\end{conjecture}

In certain cases the Space-Time Theorem (Theorem \ref{theorem:SpaceTimeTheorem}) and the lower bound as proved in Theorem \ref{theorem:lowerBound} coincide.

\begin{lemma}\label{theorem:CertainTMsHaveDimension1}
In case a TM $\tau$ uses polynomial space, and runs super-polynomial time we have that $d(\tau) = 1$.

More in general, if for a TM $\tau$ we have that $\lim_{x\to \infty} \frac{s(\tau, x)}{t(\tau,x)} = 0$, then
\[
\liminf_{x\to \infty} \frac{\log \big( s_\tau(x)\big)}{\log \big( t_\tau(x) \big)} = 0 \ \ \Longleftrightarrow \ \ d(\tau) =1.
\] 
\end{lemma}

\begin{proof}
By combining our general lower and upper bound as proven in Theorem \ref{theorem:lowerBound} and Theorem \ref{theorem:SpaceTimeTheorem} respectively, we see that
\[
1\ \leq \  d(\tau) \ \leq 1 + \liminf_{x\to \infty} \frac{\log \big( s_\tau(x) \big) }{\log \big( t_\tau(x)\big) } \leq 1.
\]
\end{proof}

Lemma \ref{theorem:CertainTMsHaveDimension1} shows us that in certain cases, the upper bound as given in the Space-Time Theorem is actually attained. We shall empirically verify that this is always the case in (3,2) space and conjecture that it holds more in general.

\begin{conjecture}[Upper Bound Conjecture]\label{conjecture:UpperBoundConjecture}
We conjecture that for each $n \in \omega$ and each TM $\tau$ in $(n,2)$ space we have $d(\tau) = 1 + \displaystyle \liminf_{x\to \infty} \frac{\log \big( s_\tau (x) \big) }{\log \big( t_\tau (x) \big) }$.
\end{conjecture}

Thus, Lemma \ref{theorem:CertainTMsHaveDimension1} provides a proof of the Upper Bound Conjecture in certain situations. For any TM that performs at most in linear time we have also proven the Upper Bound Conjecture in Lemma \ref{theorem:LinearTimeTMsHaveDimension2}. Below, in Lemma \ref{theorem:UpperBoundProof}, we shall prove the Upper Bound Conjecture for some other situations too. In order to prove this, we first need an additional insight.

\begin{lemma}\label{lemma:ConstantRatio}
For each TM $\tau$ there is a constant $c_\tau \in [0,1]$ with 
\[
\liminf_{x\to \infty} \frac{N_\tau(x)}{s_\tau (x) \cdot t_\tau (x)} = c_\tau.
\]

\end{lemma}

\begin{proof}
Since $N_\tau(x)$ is bounded above by $s_\tau (x) \cdot t_\tau (x)$ (we observed this before in the proof of Theorem \ref{theorem:SpaceTimeTheorem}) we get that for each TM $\tau$ we have for each $x$ that $\frac{N_\tau(x)}{s_\tau (x) \cdot t_\tau (x)} \in [0,1]$.
But then clearly the $\liminf$ is well defined and within the closed interval $[0,1]$.

\end{proof}

\begin{lemma}\label{theorem:UpperBoundProof}
In case $\lim_{x\to \infty} \frac{N_\tau(x)}{s_\tau (x) \cdot t_\tau (x)} \neq 0$ we can prove the Upper Bound Conjecture, that is 
\[
d(\tau) = 1 + \liminf_{x\to \infty} \frac{\log \big( s_\tau (x) \big) }{\log \big( t_\tau (x) \big)}.
\]
\end{lemma}

\begin{proof}
We may assume that $\tau$ runs in at least linear time for otherwise, the claim is proved by Lemma \ref{theorem:LinearTimeTMsHaveDimension2}. Thus $\lim_{x\to \infty} \frac{1}{t_\tau (x)} = 0$. Our assumption gives us that $\lim_{x\to \infty} \frac{N_\tau(x)}{s_\tau (x) \cdot t_\tau (x)} = c_\tau$ for some $c_\tau \neq 0$. Note that in this assumption we have a limit and not a liminf so that any subsequence converges to the same limit. Consequently we have
\[
\begin{array}{lll}
\liminf_{x\to \infty} \frac{\log (N_\tau (x))}{\log (t_\tau(x))} & = & \liminf_{x\to \infty} \frac{\log (c_\tau\cdot s_\tau (x)\cdot t_\tau (x))}{\log ( t_\tau(x))}\\\\
 & = & \liminf_{x\to \infty} \frac{\log (c_\tau) + \log( t_\tau (x)) + \log( s_\tau(x))}{\log ( t_\tau(x))}\\
\end{array}
\]
 which implies the Upper Bound Conjecture provided $c_\tau \neq 0$.
\end{proof}

The following proposition provides an almost equivalent formulation of the Upper Bound Conjecture.

\begin{proposition}\label{theorem:equivalenceUpperBoundConjecture}
For each TM $\tau$ we have that if 
\[
\liminf_{x\to \infty} \frac{\log (N_\tau(x))}{\log (s_\tau(x)t_\tau (x))} = 1,
\]
then the Upper Bound Conjecture holds.

Moreover, if the Upper Bound Conjecture holds uniformly for some TM $\tau$, that is $d(\tau) = 1 + \lim_{x\to \infty} \frac{\log (s_\tau (x))}{\log (t_\tau (x))}$, then  $\liminf_{x\to \infty} \frac{\log (N_\tau(x))}{\log (s_\tau(x)t_\tau (x))} = 1$.
\end{proposition}

\begin{proof}
If $\liminf_{x\to \infty} \frac{\log (N_\tau(x))}{\log (s_\tau(x)t_\tau (x))} = 1$, we also have 
$\lim_{x\to \infty} \frac{\log (N_\tau(x))}{\log (s_\tau(x)t_\tau (x))} = 1$ since ${\log (N_\tau(x))}\leq {\log (s_\tau(x)t_\tau (x))}$ for each $x$. Consequently,
\[
\begin{array}{lll}
d(\tau) &=& \liminf_{x\to \infty} \frac{\log (N_\tau(x))}{\log (t_\tau (x))} \\\\
 & = & \liminf_{x\to \infty} \frac{\log (s_\tau(x)t_\tau (x))}{\log (t_\tau (x))} \\\\
  & = & \liminf_{x\to \infty} \frac{ \log (t_\tau (x))+\log (s_\tau(x))}{\log (t_\tau (x))} \\\\
  & = & 1 + \liminf_{x\to \infty} \frac{\log (s_\tau (x))}{\log (t_\tau (x))} .\\
\end{array}
\]  
For the other direction we assume $d(\tau) =1 + \lim_{x\to \infty} \frac{\log (s_\tau (x))}{\log (t_\tau (x))}$. Consequently
\[
\begin{array}{lll}
\liminf_{x\to \infty} \frac{\log (N_\tau(x))}{\log (t_\tau (x))} &=&  1 + \lim_{x\to \infty} \frac{\log (s_\tau (x))}{\log (t_\tau (x))}\\\\

 &=&   \lim_{x\to \infty} \big( 1 + \frac{\log (s_\tau (x))}{\log (t_\tau (x))} \big)\\\\

 &=&   \lim_{x\to \infty} \big( \frac{\log (t_\tau (x))}{\log (t_\tau (x))} + \frac{\log (s_\tau (x))}{\log (t_\tau (x))} \big) \\\\
 
 &=&   \lim_{x\to \infty} \frac{\log (s_\tau (x)t_\tau (x))}{\log (t_\tau (x))}.
 \end{array}
\]  
Using this identity $\liminf_{x\to \infty} \frac{\log (N_\tau(x))}{\log (t_\tau (x))} = \lim_{x\to \infty} \frac{\log (s_\tau (x)t_\tau (x))}{\log (t_\tau (x))}$ we see that for any subsequence $x_n \to \infty$ we have that 
\[
\lim_{n\to \infty}  \Big(\frac{\log (N_\tau(x_n))}{\log (t_\tau (x_n))}\Big) \big/ \Big(\frac{\log (s_\tau (x_n)t_\tau (x_n))}{\log (t_\tau (x_n))}\Big) \geq 1.
\]
Consequently,
\[
\begin{array}{lllll}
\liminf_{x\to \infty} \frac{\log (N_\tau (x))}{\log(s_\tau(x) t_\tau (x))} & = & 
\liminf_{x\to \infty} \frac{\frac{\log (N_\tau (x))}{\log (t_\tau(x))}}{\frac{\log(s_\tau(x) t_\tau (x))}{\log ( t_\tau(x))}}
 & \geq  & 1. 
\end{array}
\]
But since $N_\tau(x) \leq s_\tau(x) t_\tau (x)$ the possibility $\liminf_{x\to \infty} \frac{\log (N_\tau (x))}{\log(s_\tau(x) t_\tau (x))} > 1$ cannot occur and we are done.
\end{proof}

\subsection{The Space-Time Theorem and $\sf P$ versus $\sf NP$}

As usual, we denote by $\sf P$ the class of problems that can be solved by a TM which uses an amount of time that is bounded by some polynomial applied to the size of the input (representing an instantiation of the particular problem).

Likewise, we denote by $\sf NP$ the class of problems so that any solution of this problem can be checked to be indeed a solution to this problem by a TM which uses an amount of time that is bounded by some polynomial applied to the size of the input. Here, the $\sf N$ in $\sf NP$ stands for \emph{non-deterministic}. That is to say, a non-deterministic TM would run in polynomial time by just guessing the right solution and then checking that it is indeed a solution. It is one of the major open question in (theoretical) computer science wether ${\sf P} = {\sf NP}$ or not.

By $\sf PSPACE$ we denote the class of problems that can be solved by a TM which uses an amount of memory space that is bounded by some polynomial applied to the size of the input. It is well-known that ${\sf NP} \subseteq {\sf PSPACE}$. Thus, by Lemma \ref{theorem:CertainTMsHaveDimension1} we can state a separation of $\sf P$ and $\sf NP$ in terms of dimensions.

\begin{quote}
Let $\Pi$ be some $\sf NP$-complete problem. If for each $\sf PSPACE$ Turing machine $\tau$ that decides $\Pi$ we have that $d(\tau) = 1$, then ${\sf P} \neq {\sf NP}$.
\end{quote}

Clearly, this does not constitute a real strategy since, for one, in
general it is undecidable whether $d(\tau) = 1$ \cite{Tadaki10}.

\newpage
\section*{{\LARGE Part II: Experimental setting}}
\addcontentsline{toc}{section}{Part II}

In this second part of the paper we describe the experiment we have performed to empirically test if the theoretical results also hold in cases that do not satisfy the necessary requirements for the theoretical results to be applied.

\section{The experiment}\label{section:TheExperiment}

We have already proven on purely theoretical grounds that there is a relation between runtimes and fractal dimension of the space-time diagrams. However, our theoretical results only apply to a restricted class of TMs.

In the experiment we wanted to also study the fractal dimension of the space-time diagrams in cases where our theoretical results do not apply. Moreover, guided by the first outcomes of our experiment we formulated the Upper Bound Conjecture (Conjecture \ref{conjecture:UpperBoundConjecture}) and gathered data so to investigate if the conjecture holds in $(3,2)$ space.

\subsection{Slow convergence}
For TMs $\tau$ that run in at most linear time we have proven in Lemma \ref{theorem:LinearTimeTMsHaveDimension2} that $d(\tau) = 2$. Our aim is to use computer experiments to compute the Box dimension of all TMs $\tau$ where $d(\tau)$ is not predicted by any theoretical result.

A substantial complication in this project is caused by the occurrence of logarithms in the definition of $d(\tau)$. As a consequence, increase in precision of $d(\tau)$ requires exponentially larger inputs. This makes direct brute-force computation unfeasible. As an example, let us consider 2,2-TM 346 again whose space time diagrams we saw in Figure~\ref{figure:TM346SpaceTimeDiagrams}. By Lemma \ref{theorem:LinearTimeTMsHaveDimension2} we know that the Box dimension of this Turing machine equals two. However, Figure \ref{figure:TM346SlowConvergenceBoxDimension} below shows us how slow the rate of convergence is.

\begin{figure}[htb!]
  \centering
  \includegraphics[height=8cm, angle =-90]{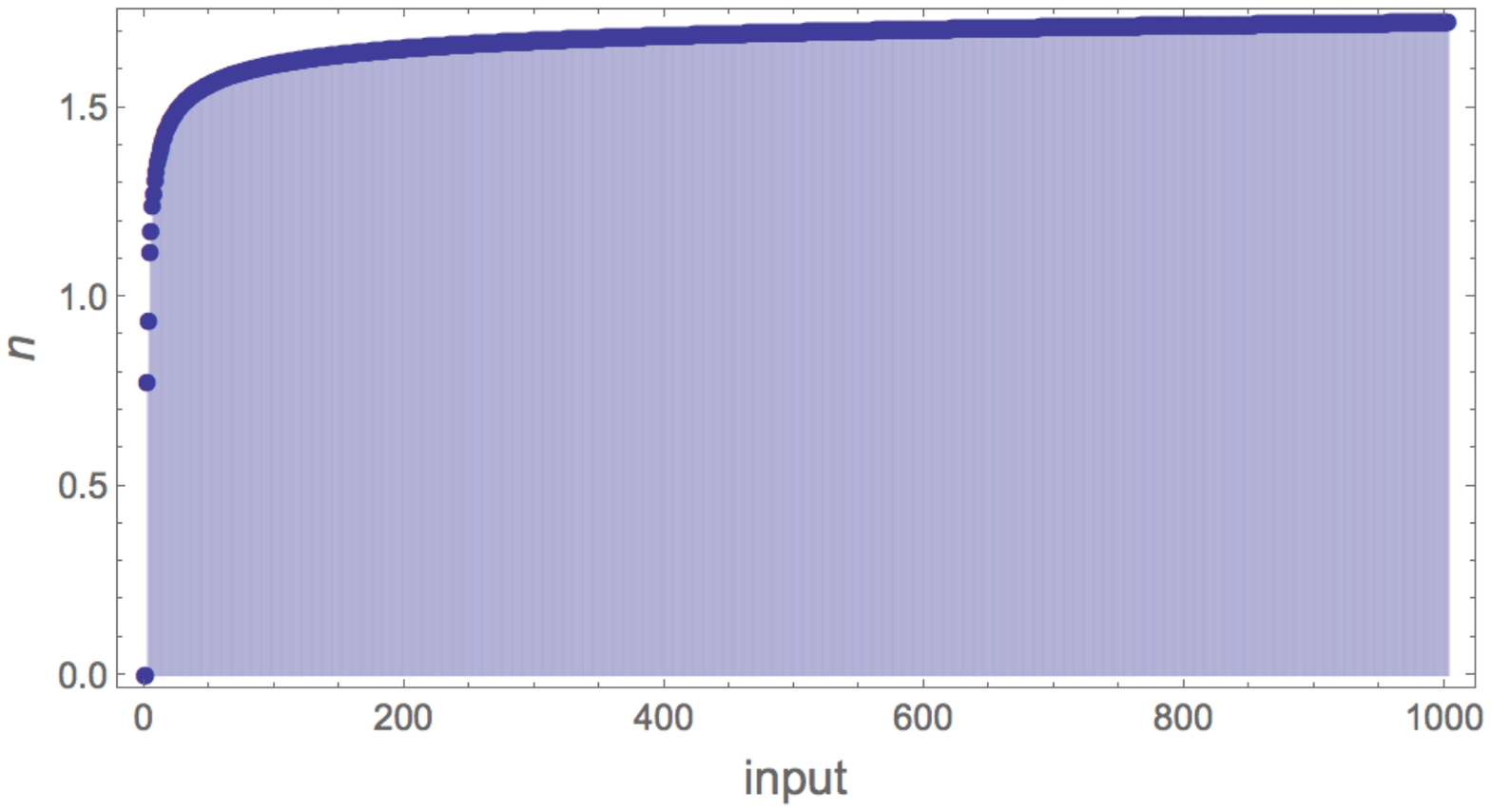}
  \caption{The figure shows an estimate of the Box dimension of 2,2 TM with TM number 346. On the horizontal axis the input is shown and the vertical axis shows the corresponding approximation of the Box dimension. Note that we know that the function converges to 2 when the input tends to infinity.}
  \label{figure:TM346SlowConvergenceBoxDimension}
\end{figure}

Our way out here is to apply numerical and mathematical analysis to the functions involved so that we can retrieve their limit behavior. In particular, we were interested in three different functions.

As before, for $\tau$ a TM we denote by $t_\tau(x)$ the amount of time-steps needed for $\tau$ to halt on input $x$; by $N_\tau(x)$ we denote the number of black cells in the space-time diagram of $\tau$ on input $x$ and by $s_\tau(x)$ the distance between the edge of the tape and the furthest cell visited by $\tau$ on input $x$.

With these functions and knowledge of their asymptotic behavior, we can compute the corresponding dimension $d(\tau)$ and the upper bound\\ $1~+~\liminf_{x\to \infty}\frac{s_\tau (x)}{t_\tau (x)}$. The functions are guessed by looking at large enough initial sequences of their outcomes in a mechanized fashion. The few cases that cannot be done in a mechanized version were analyzed by hand.

It is important to bear this process in mind and the fact that we work with guesses that can be wrong in principle. For example, if we speak of a TM $\tau$ that performs in time of order $n^2$ this means in this paper that,  by definition, after applying our particular analyzing process, $\tau$ was classified as an $\mathcal{O} (n^2)$ time performer. It may well be that in reality $\tau$ needs exponential time. However, there are strong indications that our guessing process is rather accurate \cite{ZenilSJMindsMach12,JoostenSZ11}.

\subsection{Methodology}

In this subsection we shall describe the steps that were performed in obtaining our results.
Basically, the methodology consists of the following steps.

\begin{enumerate}
\item
Each TM that lives in 2,2 space also occurs in (3,2) space so for the final results it suffices to focus on this data-set. The TMs that diverge on all inputs were removed from the initial list of 2\,985\,984 TMs in the (3,2) space, since for them the dimension is simply not defined. For the remaining TMs we erased all diverging inputs from the sequence to which we were to apply our analysis. Since we are only interested in limit behavior of any subsequences this does not alter our final results.

We isolated the TMs for which there is no theorem that predicts the corresponding dimension. By Lemma \ref{theorem:LinearTimeTMsHaveDimension2} and Lemma \ref{theorem:CertainTMsHaveDimension1} this means that we only needed to pay attention to those TMs which use more than linear time. Moreover, we also removed all simultaneous EXP time and PSPACE performers to finally end up with a collection of TMs. The distribution of the resulting collection is summarized in Table \ref{tab:comClasses32} below.

\begin{table}[htbp!]
  \centering
  \begin{tabular}{lllr}
    \textbf{Boxes} & \textbf{Runtime} & \textbf{Space} &
    \textbf{Machines}\\
    $\mathcal{O}(n^3)$ & $\mathcal{O}(n^2)$ & $\mathcal{O}(n)$ & 3358 \\
    $\mathcal{O}(n^4)$ & $\mathcal{O}(n^3)$ & $\mathcal{O}(n)$ & 6 \\
    $\omega({\sf P})$ & $\omega({\sf P})$ &  $\omega({\sf P})$ & 14 \\
  \end{tabular}
  \caption{Distribution of those TMs in (3,2) space of which we had to compute the corresponding dimension over their complexity classes. By $\omega({\sf P})$ we denote the little $\omega$ notation of the class of polynomials and hereby collect any super-polynomial behavior in one bucket.}
  \label{tab:comClasses32}
\end{table}

In addition there are 1\,792 TMs that perform in exponential time and linear space, but clearly they needed no further analysis since we know on theoretical grounds that their corresponding dimension is 1. All other machines in (3,2) space were very simple in terms of time computational complexity, that is, they perform at most in linear time.

\item
Per TM $\tau$, we determined/guessed its function $s_\tau(x)$ corresponding to the space usage of $\tau$ on input $x$. Although this guessing was already performed in \cite{JoostenSZ11} we decided to re-do the process. The main reasons to do this were a new release of our analyzing tool \emph{Mathematica} together with the fact that the authors had obtained new insights on how to best perform the analysis. Our results coincided in large part with the ones obtained in \cite{JoostenSZ11} but also showed minor discrepancies.

\item
Per TM $\tau$, we determined its function $t_\tau(x)$ corresponding to the time usage of $\tau$ on input $x$;

\item
Per TM $\tau$, we determined its function $N_\tau(x)$ corresponding to the number of black cells in the space-time diagram of $\tau$ on input $x$;

\item
Per TM $\tau$, we computed $\liminf_{x\to \infty} \frac{s_\tau(x)}{t_\tau(x)}$;


\item
Per TM $\tau$, we computed its dimension $d(\tau)$ as $d(\tau) = \displaystyle \liminf_{x\to \infty} \frac{\log(N_\tau(x))}{\log(t_\tau(x))}$;

\item
Per TM $\tau$, we compared its dimension $d(\tau)$ to its theoretical upperbound $1+ \displaystyle \liminf_{x\to \infty} \frac{\log(s_\tau(x))}{\log(t_\tau(x))}$ which we computed separately;

\item
Per TM $\tau$, we computed $\liminf_{x\to \infty} \frac{\log (N_\tau(x))}{\log (s_\tau(x)t_\tau (x))}$ and $\liminf_{x\to \infty} \frac{N_\tau(x)}{s_\tau (x) \cdot t_\tau (x)}$.

\end{enumerate}

\subsection{Alternating convergent behavior}
Some of the Turing Machines possessed alternating asymptotic behavior. This has been already observed in \cite{JoostenSZ11}. Typically the alternation reflects modular properties of the input like being odd or even or of the number of states.

The differences between the alternating subsequences can be rather drastic though. The most extreme example we found is reflected in Figure \ref{fig:alternatingBeh} below.

\begin{figure}[htbp!]
  \centering
  \includegraphics[height=8cm,angle=90]{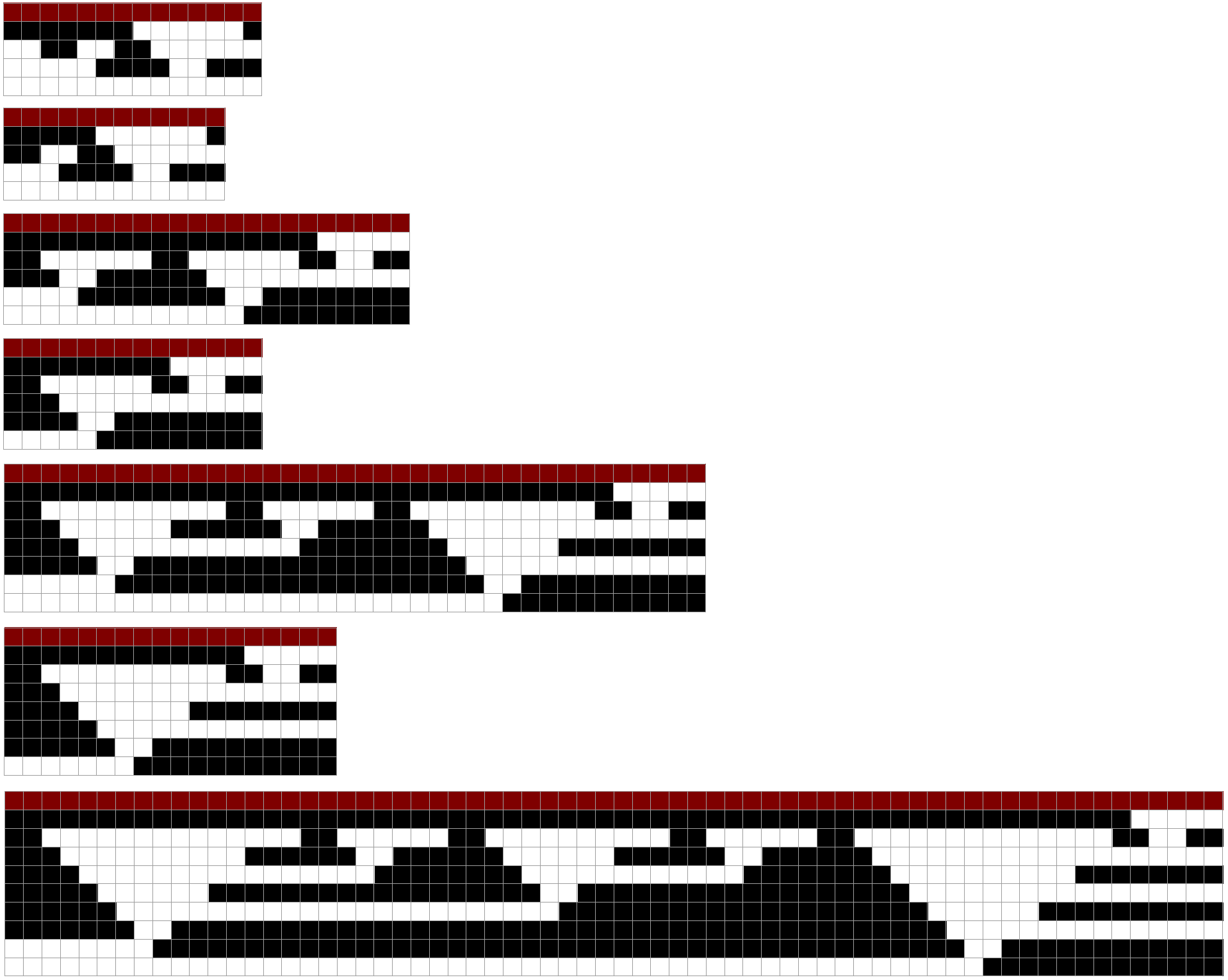}
  \caption{Alternating linear and exponential runtime behavior for TM 1\,728\,529}
  \label{fig:alternatingBeh}
\end{figure}

Figure \ref{fig:alternatingBeh} shows the space-time diagrams for TM $\tau$ with number
1\,728\,529 for inputs 1 to 7. For convenience we have changed the orientation of the diagrams so that time `goes from left to right' instead of from `top to bottom'.

This machine runs in linear time for even inputs and
exponential time for odd inputs. The runtime is given by:
\begin{displaymath}
  t_\tau(x) = \left\{
      \begin{array}{ll}        
        2 (x-2)+9 & \text{if} \ x \ \text{is even;}\\
        2 (x-1)+3\ 2^{\frac{x-1}{2}+1}+5 & \text{if} \ x \ \text{is odd.}
      \end{array}\right.
\end{displaymath}
The number of black cells ($N_\tau (x)$) in the space-time diagram exhibits the same behavior. Note however, that the space that $\tau$ uses  is linear in the size of the input and in particular the amount of tape cells used is equal to the size of the output.

Moreover, we note that the sequence of outputs is of a very simple and regular nature. The outputs can be grouped in series of two, where the output on input $2\cdot n +1$ consecutive black cells is equal to the output on input $2\cdot n + 2$ consecutive black cells. So, in a sense this TM incorporates two different algorithms to compute this output: one in linear time, the other, in exponential time.

We have found alternating sequences of periodicity 2, 3 and 6. 
Like we noted in \cite{JoostenSZ11}, the periodicity typically reflects either the number of states, the number of colors, or a divisor of their product. 
Figure~\ref{fig:alternating6} shows an example of TM Number 1\,159\,345 whose corresponding box counting sequence $N_\tau (x)$ has periodicity six.

\begin{figure}[htbp!]
  \centering
  \includegraphics[width=8.1cm,angle =-90]{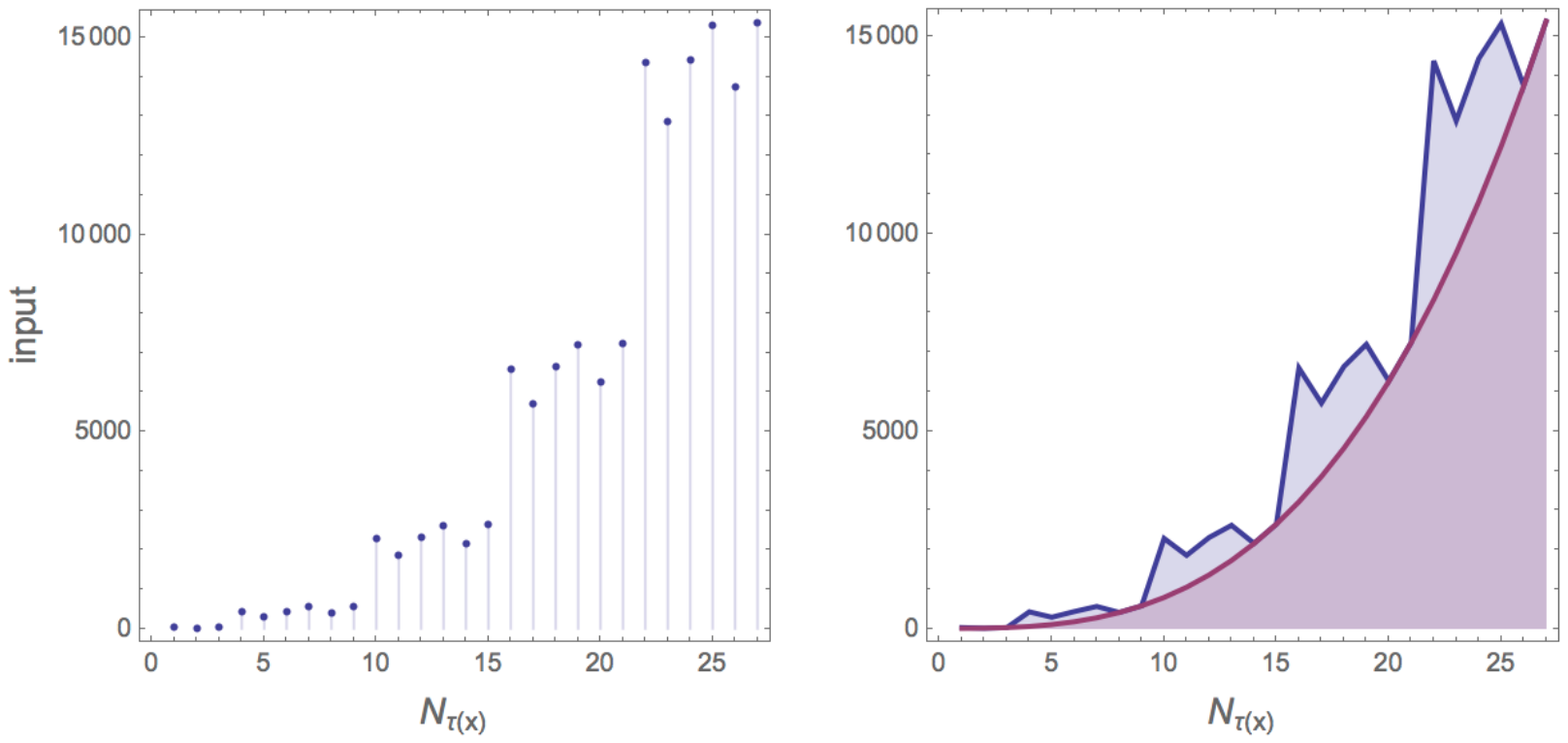}
  \caption{Alternating values for $N_\tau(x)$ with periodicity 6. The depicted values are for TM number 1\,159\,345. The diagram on the left shows the data-set and on the right we included a fit from below.}
  \label{fig:alternating6}
\end{figure}

On the left of Figure~\ref{fig:alternating6} we show the points $N_\tau(x)$ on the vertical axis plotted against the input $x$. On the right of the same figure we estimated a fit from below.

This alternating behavior reflects the richness of what we sometimes refer to as \emph{the micro-cosmos of small Turing machines}. It is this alternating behavior which complicated analyzing the data set in a straight-forward automated fashion.

\subsection{Determining the important functions}

In this subsection we would mainly like to stress that most of the computational effort for this paper has actually been put into determining/guessing the functions $s_\tau (x), t_\tau(x)$ and $N_\tau(x)$ and computing the corresponding limits.

As may have become manifest from the previous subsection, it is hard to automatically guess perfect matches for these functions in case there is alternating behavior present. Finally we could deal with all functions in a satisfactory way. Notwithstanding our confidence, it is good to bear in mind that all classifications provided in this paper are classifications given the current methodology.

We shall here briefly describe how we proceeded to guess our functions. The methodology is fairly similar as performed in \cite{JoostenSZ11}. However, for this project we used newer tools and a slightly more sophisticated methodology which accounts for possible differences with \cite{JoostenSZ11}. Schematically the guessing process can be split into the following steps.

\begin{enumerate}
\item
We collected the sequences for time-usage $t_\tau (x)$ and space-usage $s_\tau(x)$ from the TM data set as described in Section \ref{section:SmallTMdataBase} of this paper.

\item
These sequences $t_\tau (x)$ and $s_\tau(x)$ are only given for the first 21 different inputs.
We used an initial segment of 15 elements of these sequences to guess
in an automated fashion the corresponding function that allegedly
generates this sequence. In some cases the beginning
of the sequence (up to three elements) was removed because the beginning did not match the general pattern that only occurred later on in the sequence. If we would leave the first values, \emph{Mathematica} was no longer able to find the general pattern.
The guessing process was done in \emph{Mathematica} v.8 and 9 using the \verb;FindSequenceFunction;
as built-in in this software. In some cases 
\verb;FindSequenceFunction; came with a solution, in others it did not. The function
\verb;FindSequenceFunction; does various standard numerical and algebraic analyses on the sequences but also checks for obvious recurrence patterns. The function, built into the computer algebra system \emph{Mathematica}, takes a sequence of integer values $\{a_1, a_2, \ldots, a_m\}$ to define a function that yields a sequence $\{a_n\}_{n\in \omega}$ which coincides on the first $m$ values.

\verb;FindSequenceFunction; finds results in terms of a wide range of integer functions such as sums and series coefficients, as well as implicit solutions to difference equations using early elements in the list to find candidate functions, then validates the predicted function by looking at later elements.

\item
Thus, we obtain two lists: a list $L_1$ of TMs where we found a guess
and a list $L_2$ where we did not find any guess. From the initial list of 528 runtime sequences, we
  could not guess 11, and from the 167 space sequences, we could not guess
  15. Note that this number need not be equal since TMs from various different functions had the same space sequence. 
  
  Moreover, 288 runtime sequences and 85 space sequences in $L_1$ were
alternators. \emph{Mathematica} guessed the right function using terms like $(-1)^x$. However, in computing the $\liminf$, we manually split those sequences into, for example, an even and an odd part, to obtain the corresponding limits.

\item
We performed a check on our guesses as collected in $L_1$ by applying
the guessed function to inputs 16--21. In almost all cases our guess
turned out to be predictive and coincided with the real values. For
those few cases where there was a discrepancy between the guesses and
the actual values, we made a new guess based on a larger initial segment, now consisting of the first 18 elements and then testing it once more on new real values. Finally we were able to guess and successfully check all of the sequences --both space and time usage-- in $L_1$.

\item
From the list $L_1$ we deleted all complexities for which we knew the dimension on theoretical grounds so to obtain a list $L_3$.

\item
For the TMs in $L_3$ we used the supercomputing resources of CICA
(Andalusian Centre for Scientific Computing) to compute the
corresponding sequences $N_\tau (x)$ with a C++
  TM simulator. To reduce the computational effort, for each set of equivalent TMs (up to a geometrical transformation, such as state mirroring) only one representative was run. We applied the guessing process as described above for $t_\tau(x)$ and $s_\tau(x)$
also to $N_\tau(x)$ to come up with corresponding functions.

\item
For the sequences in $L_2$ we applied a semi-manual process. Basically, there were three different procedures that we applied so to find solutions also in $L_2$ for the sequences $s_\tau(x), t_\tau(x)$ and $N_\tau(x)$.
 
\begin{enumerate}
\item
In most of the cases, there was alternating behavior present. We could
read off the periodicity from looking at graphs as for example in
Figure \ref{fig:alternating6}. Sometimes, looking directly at the
space-time diagrams was more informative. In all of these cases but one, we
finally did find functions for the subsequences using our methodology
as described above. As splitting the sequences into 2, 3 or 6
  alternating ones reduces the length of the input sequence of
  \texttt{FindSequenceFunction}, we run in some cases 40 or 60 more inputs
with the C++ simulator to end up with a sufficiently large data set.

One alternating TM did not succumb to this methodology. This was TM 582\,263 whose treatment is included in Section \ref{sec:space-time-theorem}. We run this TM for 35 inputs with the simulator and
observed that $N_\tau(x)/(s_\tau(x) t_\tau(x))$ clearly converges to a
constant --one for each subsequence-- so we approximated $N_\tau(x)$ by $c\cdot  s_\tau(x) t_\tau(x)$
which was enough for the log-limit without knowing the exact value of
$c$.

\item
In some cases the regularity was not obvious to \emph{Mathematica} but where evident when looking at space-time diagrams and/or the binary expansion of the output. In these cases we could manage by just feeding our insight into \emph{Mathematica} in that we let it work, for example, on the binary expansion of the sequences.

\item
In some cases the recurrences were just too complicated for \emph{Mathematica} v8. In these cases we carefully studied the space-time diagrams analyzing what kind of recurrences there were present. Then, the observed recurrences were fed into \texttt{FindSequenceFunction} where we left FindSequenceFunction find out the exact nature and coefficients of the corresponding recurrences. One such example concerns the TM that produces the largest possible outputs in (3,2) space: the so-called \emph{Busy Beaver} as detailed in Section \ref{section:ExpSpaceAndBusyBeaver}.

\end{enumerate}

\item
After having successfully (allegedly) found the functions $s_\tau(x)$, $t_\tau(x)$ and $N_\tau (x)$ we could compute the values for $d(x) = \liminf_{x\to \infty} \frac{N_\tau(x)}{t_\tau(x)}$ and $\liminf_{x\to \infty} \frac{s_\tau(x)}{t_\tau(x)}$. In most cases a simple limit sufficed. For alternating behavior we had to select most of the times the subsequences by hand so to end up with the $\liminf$ value. For some alternating sequences the $\liminf$ value could just be obtained by combining on the one hand the $\liminf$ of $N_\tau(x)$ (as depicted in Figure \ref{fig:alternating6}) or $s_\tau(x)$ respectively, and on the other hand $t_\tau (x)$.

\end{enumerate}

\section{Most salient results of the experiment}\label{section:MostSalientFindings}

In this section we shall present the main results of our investigations. 
The space of TMs which employ only  2 colors and 2 states is clearly contained in (3,2) space. However, we find it instructive to dedicate first a subsection to the findings in (2,2) space. Apart from the first subsection, all other results in this section refer to our findings in (3,2) space.

\subsection{Findings in (2,2) space}

In (2,2) space there was a total of 74 different functions. Of these functions, only 5 of them where computed by some super-linear time TMs. Note, this does not mean that all TMs computing this function performed in super-linear time. For example, the tape identity has many constant time performing TMs that compute it, but also some exponential time performing TMs that compute it.

In total, in (2,2) space, there  are only 7 TMs that run in super-polynomial time. Three of them run in exp-time, all computing the tape-identity. The other four TMs compute different functions. These functions do roughly compute a function that doubles the tape input, see Figure \ref{figure:FourQuadraticPerformers}.

\begin{figure}[htb!]
  \centering
  \includegraphics[height=5.7cm]{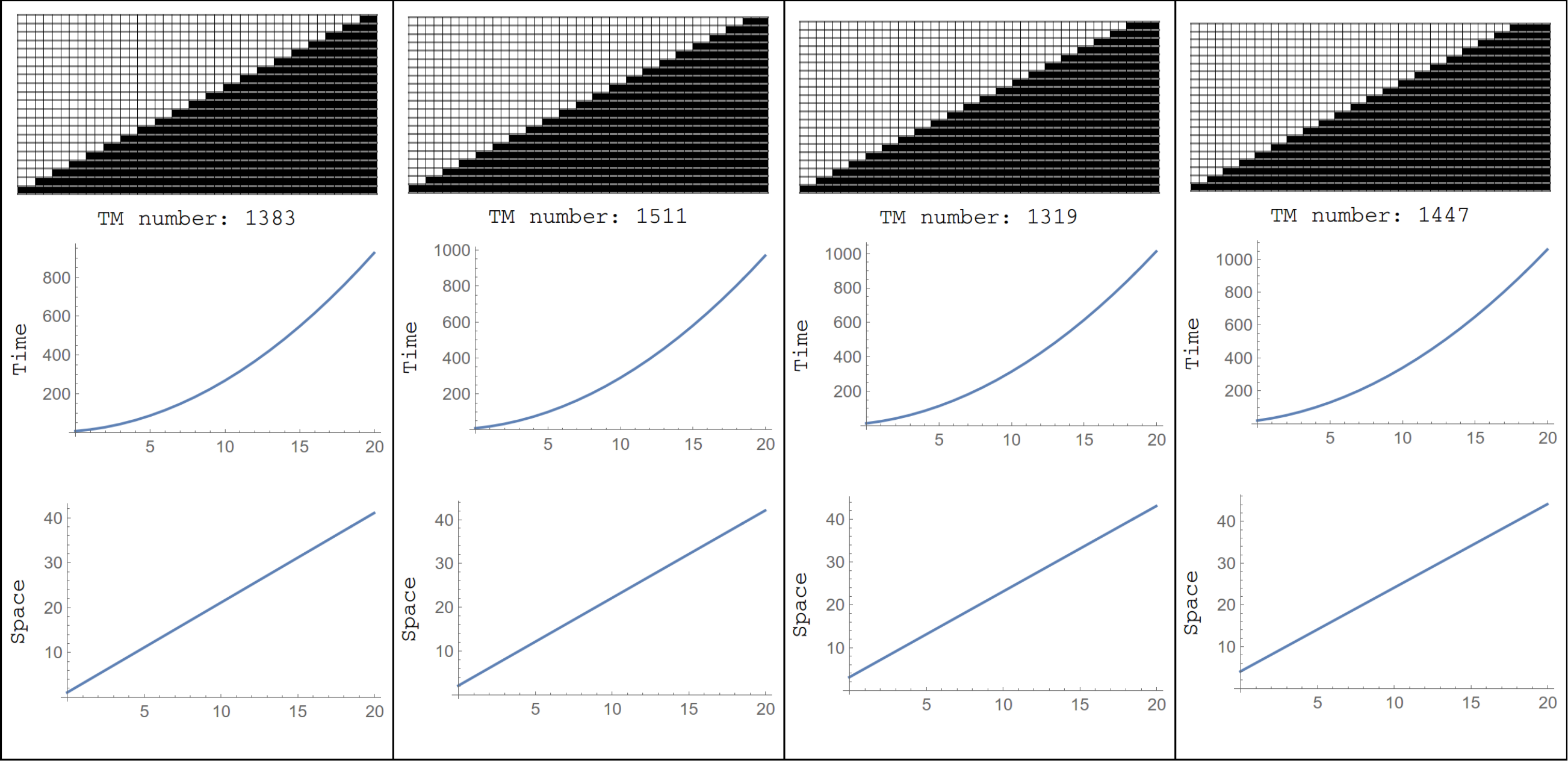}
  \caption{The figure shows the four different functions that are computed by the four TMs that have quadratic runtime in 2,2 space. The diagrams show the outputs on increasing inputs. So for example, in the left-most diagram we see that TM with number 1383 (recall this is the code in (2,2) space) outputs two black consecutive cells on input 1, and more in general it outputs $2n$ black consecutive cells on input $n$.}
  \label{figure:FourQuadraticPerformers}
\end{figure}

All these four TMs perform in quadratic time and linear space. We computed the dimension for these functions and all turned out to have dimension $\frac{3}{2}$. We observe that this is exactly the upper bound as predicted by the Space-Time Theorem. We saw this phenomenon in (3,2) space as well.

The only three exponential time performers used linear space so by Lemma \ref{theorem:CertainTMsHaveDimension1} we already know that the dimension of those TMs should be one. This has been checked also in \emph{Mathematica}. The check was not really performed to check our theoretical results, rather the check was used as a test-case for our analyzing software.

We saw that a TM in (2,2) space runs in super-polynomial time if and only if its dimension equals 1. This observation is no longer valid in (3,2) space though.

\subsection{Exponential space and the Busy Beaver}\label{section:ExpSpaceAndBusyBeaver}
%
In the remainder of this section we shall focus on the TMs in (3,2) space. That space contains 2\,985\,984 many different TMs which compute 3\,886 different functions. Almost all TMs used at most linear space for their computations. The only exception to this was when the TM used exponential space. Curiously enough, in (3,2) space there was no space usage in between linear and exponential space.

In \cite{OurNewDemonstration} one can see an overview of the EXP-SPACE performing TMs. For most of these TMs it was not too hard to find an explicit formula for the space usage. An example is TM with number 683\,863 whose corresponding space usage is:
\begin{displaymath}
s_{683.863}(x) = 2 \left(\frac{x+1}{2}+2^{\frac{x+1}{2}}-1\right)  
\end{displaymath}
%
The space-time time diagrams for TM
683\,863 contained sufficiently much regularity so that \emph{Mathematica} could guess the corresponding functions. For various other EXP-SPACE performers we had to help \emph{Mathematica} by suggesting it to what kind of recursion it should look for. This occurred also with the so-called \emph{Busy Beaver}.

Classically speaking the \emph{Busy Beaver function} outputs on input $n$ the longest time that any TM with $n$ states runs when executed on a two-way infinite tape with empty input \cite{Rado1962}. In analogy, in the context of this paper we shall call a TM $\beta$ a \emph{Busy Beaver} whenever for each TM $\tau$, there is some value $x_0$ so that for all $x\geq x_0$ we have $t_\beta(x) \geq t_\tau (x)$. 
The equivalent machines 599\,063 and 666\,364 are the Busy Beavers in the (3,2) space. They
compute the largest runtime, space and boxes sequences. They also produce the longest output strings. For the remainder of this subsection we shall denote the Busy Beaver TM by $\beta$. As mentioned, there are of course two actual TMs that compute the Busy Beaver but they have the exact same behavior and we shall not distinguish between them.

\begin{figure}[htbp!]
  \centering
  \includegraphics[height=12cm,angle=90]{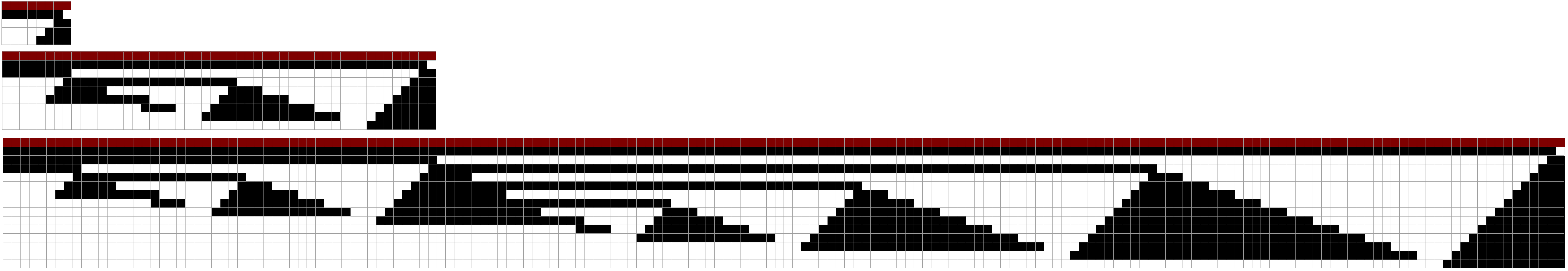}
  \caption{Execution of the Busy Beaver on the first three inputs}
  \label{fig:BBexec}
\end{figure}

Figure~\ref{fig:BBexec} shows the execution of machine 666\,364 for 
inputs 1 to 3. The diagrams have been rotated to save space. As one can see, the series of outputs is very regular and so is the sequence of cells used by the computation.  Nonetheless, \emph{Mathematica} did not find a recurrence between the consecutive values. This was due to a minor error term.

That is, if one looks at the amounts of used cells for consecutive inputs and their differences, then modulo a small error term, there is a clear tendency. Let $x$ denote the number of consecutive black input cells. Looking at the ratio between consecutive values yielded us to isolate the disturbing difference term.

\begin{table}[htbp!]
  \centering
  \begin{tabular}{|c|c|c|c|c|}\hline
    $\mathbf{x}$   & $\mathbf{s_\beta(x)}$ & $\mathbf{s_\beta(x)-s_\beta(x-1)}$ & $\mathbf{3/2(s_\beta(x-1) -
      s_\beta(x-2))}$ & \textbf{Difference}\\\hline\hline 
1&    3 & -- & -- &  -- \\
2&    7 & 4  & -- &  -- \\
3& 13 & 6 & 6 & 0\\
4&    22 & 9 & 9 & 0\\
5&    36 & 14 & 13+$1/2$ & $1/2$ \\
6&    57 & 21 & 21 & 0\\
7&    88 & 31 & 31+$1/2$ & $-1/2$ \\
8&   135 & 47 & 46+$1/2$ & $1/2$ \\
9&   205 & 70 & 70+$1/2$ & $-1/2$ \\
10&   310 & 105 &  105 & 0\\\hline
  \end{tabular}
  \caption{The structure of the space sequence}
  \label{tab:spaceBB}
\end{table}

So, ignoring the exact nature of the error term, the recurrence equation for the space is given in~\eqref{eq:1}.

\begin{equation}
  \label{eq:1}
    \begin{split}
    s_{\beta}(1) = & 3\\
    s_{\beta}(2) = & 7\\
  s_{\beta}(p) = & \frac{1}{2} \left(5 \ s_\beta(p-1) - 3 \ s_\beta(p-2) + g(p)\right) \\
  \end{split}
\end{equation}
where $g(p)$ is a function\footnote{When we forced \emph{Mathematica} to focus on the error term, it came up with the exact recurrence where $g(p)= \bigg(-(3 \ s_\beta(p-3)-5 \ s_\beta(p-2)+2
\ s_\beta(p-1))^3
 -2  \big(\frac{3}{2} \ s_\beta(p-4)-\frac{5}{2} \ s_\beta(p-3)+
\ s_\beta(p-2)\big)   \big(1-(3 \ s_\beta(p-3)-5
   \ s_\beta(p-2)+2 \ s_\beta(p-1))^2\big)+ 
\big(1-(3 \ s_\beta(p-4)-5 \ s_\beta(p-3)+2 \ s_\beta(p-2))^2\big)
    \big(1-(3 \ s_\beta(p-3)-5 \ s_\beta(p-2)+2
   \ s_\beta(p-1))^2\big)\bigg)
    \sin ^2\big(\pi 
 (\frac{5}{2}
   \ s_\beta(p-1)-\frac{3}{2} \ s_\beta(p-2))\big)$ after defining $s_\beta(-1)=s_\beta(0)=0$.} that takes values in
$\{-1, 0, 1\}$.

The runtime depends on the space and we found the following
recurrence relation for it:
\begin{equation}
  \label{eq:3}
  \begin{split}
  t_\beta(1) = & 7 \\
  t_\beta(i) = & \frac{3}{2} \ s_\beta(i-1)^2-\frac{1}{2} \sin
  ^4\left(\frac{1}{2} \pi  \ s_\beta(i-1)\right)+ \\ 
 & \frac{1}{2} \ s_\beta(i-1) (\cos (\pi \ s_\beta(i-1))+15)+ \ t_\beta(i-1)+8
  \end{split}  
\end{equation}

\noindent
Finally, by close inspection on the space-time diagrams we could guide \emph{Mathematica} to look for specific kind of recurrences to finally come up with

\begin{equation}
  \label{eq:4}
  \begin{split}
    N_\beta(1) = & 13\\
    N_\beta(i) = &  \frac{1}{32} \big(32 \ N_\beta(i-1)+32 \ t_\beta(i)+32 \
      s_\beta(i-1)^3 \\
      & + 152 \ s_\beta(i-1)^2+140 \ s_\beta(i-1)+ 16
   \ s_\beta(i)^2 \\ & + 16 \ s_\beta(i)+16 \ s_\beta(i-1)^2 \cos (\pi  \
   s_\beta(i-1)) \\
& + 24 \ s_\beta(i-1) \cos (\pi  \ s_\beta(i-1)) \\ & - 4
   \ s_\beta(i-1) \cos (2 \pi  \ s_\beta(i-1))-3 \cos (\pi  \ s_\beta(i-1))
   \\ & -9 \cos (2 \pi  \ s_\beta(i-1))-\cos (3 \pi 
   \ s_\beta(i-1))+32 i-19\big)  
  \end{split}  
\end{equation}

Using these recurrence equations we could finally compute the limits. We computed the limits both by standard methods on limits of recurrence relations and by employing \emph{Mathematica} and both methods gave the same answers to the effect that all simultaneous EXP-TIME and EXP-SPACE TMs in (3,2) space all have fractal dimension $\frac{3}{2}$.

\subsection{The space-time theorem revisited}
\label{sec:space-time-theorem}

One of our most important empirical findings is that the upper bound as given by the Space-Time Theorem is actually always attained in (3,2) space. Moreover we found two related empirical facts for (3,2) space. We mention them in this section.

\begin{quote}{\bf Finding 0.} For all TMs $\tau$ in (3,2) space we found that $d(\tau) \geq  1$. More in particular, we found that for each TM $\tau$ in (3,2) space which performed in super-linear time we have
\[
\lim_{x\to \infty} \frac{s(\tau, x)}{t(\tau, x)} = 0.
\]
and we conjecture that this holds in general for TMs with a larger number of states. 
\end{quote}

\begin{quote}{\bf Finding 1.} For all TMs $\tau$ in (3,2) space we found that
\[
d(\tau) =  1 +\liminf_{x \to  \infty} \  \frac{\log(s_\tau(x))}{\log(t_\tau(x))}
\]
and we conjectured in the Upper Bound Conjecture (\ref{conjecture:UpperBoundConjecture}) that this holds in general for TMs with a larger number of states. 
\end{quote}

In Proposition \ref{theorem:equivalenceUpperBoundConjecture} we saw that a sufficient condition for the Upper Bound Conjecture to hold is that $\liminf_{x\to \infty} \frac{\log (N_\tau(x))}{\log (s_\tau(x)t_\tau (x))} = 1$ but it is not known if this is also a necessary condition. The following finding is related to this.

\begin{quote}{\bf Finding 2.} For all TMs $\tau$ in (3,2) space we found that
\[
\liminf_{x\to \infty} \frac{\log (N_\tau(x))}{\log (s_\tau(x)t_\tau (x))} = 1.
\]
\end{quote}

In Lemma \ref{theorem:UpperBoundProof} it was shown that $\lim_{x\to \infty} \frac{N_\tau(x)}{s_\tau (x) \cdot t_\tau (x)} \neq 0$ is a sufficient condition for the Upper Bound Conjecture to hold but it is not known if it is also necessary. The following finding is related to this.

\begin{quote}{\bf Finding 3.} For all TMs $\tau$ in (3,2) space we found that
\[
\lim_{x\to \infty} \frac{N_\tau(x)}{s_\tau (x) \cdot t_\tau (x)} \in (0,1]
\]
if this limit was well defined. Thus, in particular, we found that $\lim_{x\to \infty} \frac{N_\tau(x)}{s_\tau (x) \cdot t_\tau (x)} \neq 0$. Moreover, we found that only a limited amount of numbers were attained as limits of this quotient. The values found in $(3,2)$ for $\lim_{x\to\infty}\frac{N_\tau(x)}{s_\tau(x)
  t_\tau(x)}$ are:
\[
\frac{1}{9},\frac{1}{6},\frac{7}{30},\frac{1}{4},
\frac{5}{18},\frac{5}{16},\frac{1}{3},\frac{3}{8},\frac{8}{21}, 
\frac{7}{18},\frac{5}{12},\frac{3}{7},\frac{4}{9},\frac{7}{15},
\frac{1}{2},\frac{5}{9},\frac{9}{16},\frac{2}{3},\frac{3}{4},\frac{7}{9},1
\]
It is possible that a few other limit values exist but were not found
by the way we computed the functions generating $N_\tau(x)$.
\end{quote}

\begin{figure}[htbp!]
  \centering
  \includegraphics[height=12cm, angle=90]{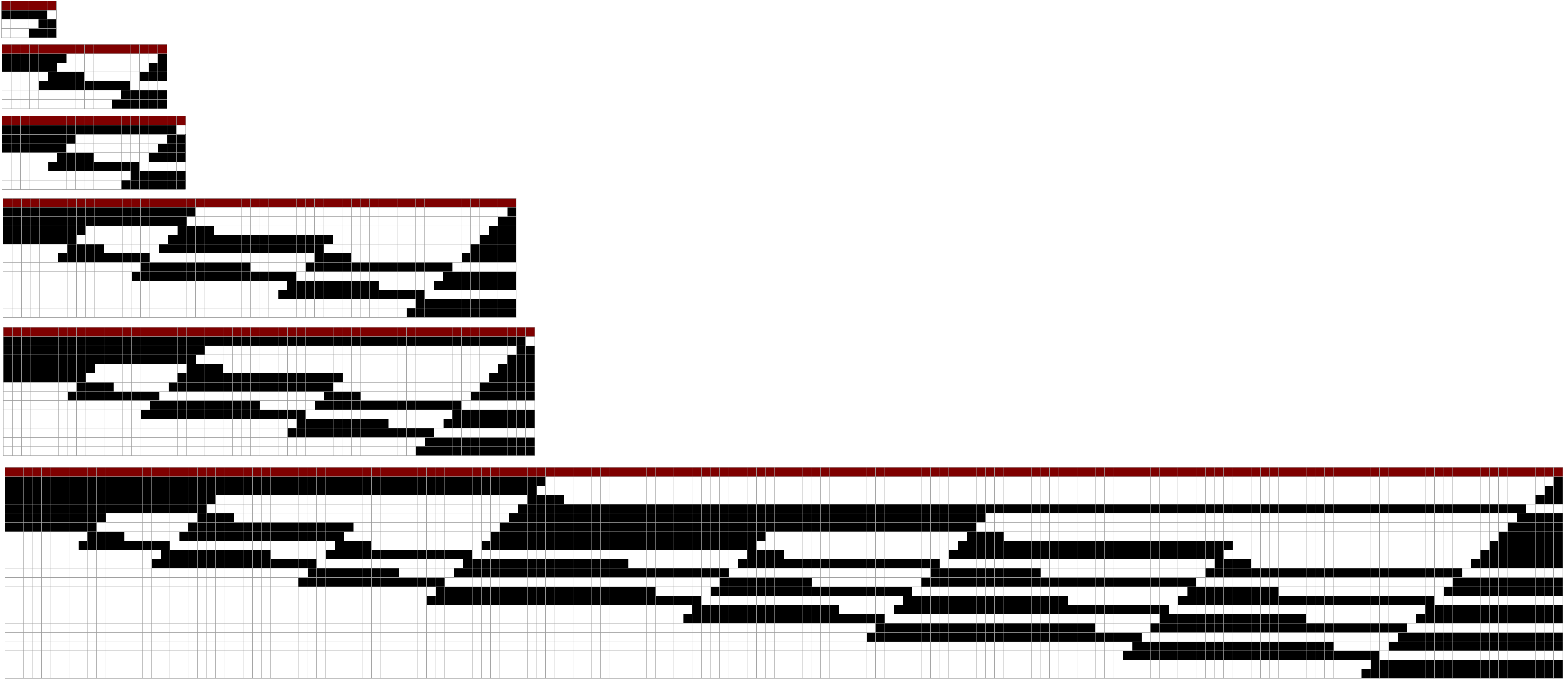}
  \caption{Execution of machine 582\,263 on the first six inputs and their corresponding space-time diagrams}
  \label{fig:relBoxesSTExec}
\end{figure}


For two of the exp-space performers we couldn't find the boxes
function. These TMs were 
582\,263 (and its twin machine), whose execution for inputs 1 to 6 is
shown in Figure~\ref{fig:relBoxesSTExec}. This TM possesses alternating
behavior with periodicity two. For these two machines we used Lemma \eqref{lemma:ConstantRatio} to settle the computation of $d(\tau)$.

For the sequences of even number of consecutive black input cells we found that the fraction 
$\frac{N(x)}{s_\tau (x) \cdot t_\tau (x)}$ tended to 0.31 whereas for the odd number of consecutive black input cells we saw it tended to 0.11. The exact value of the fraction is of course irrelevant in the computation of the limit that determines $d(\tau)$.

\begin{quote}{\bf Finding 4.} For all TMs $\tau$ in (3,2) space we found that $d(\tau) =  1$ if and only if the TM ran in super-polynomial time using polynomial space. We suspect that this equivalence holds no longer true in higher spaces, i.e., spaces $(n,2)$ for $n>3$.
\end{quote}

\begin{quote}{\bf Finding 5.} For all TMs $\tau$ in (3,2) space we found that $d(\tau) =  2$ if and only if the TM ran in at most linear time. It is unknown if this equivalence holds true in higher spaces. Note that the if part holds in general and is proven in Lemma \ref{theorem:LinearTimeTMsHaveDimension2}.
\end{quote}

\subsection{Richness in the microcosmos of small Turing machines}

The authors have explored the space of small Turing machines before. On occasion they have been so much impressed by the rich structures present there that they came to speak of \emph{the microcosmos of small Turing machines}. For this paper we had to mine (3,2) space even further and at some point were surprised to be surprised once more.

\begin{figure}[htbp!]
  \centering
  \includegraphics[width=8cm,angle =-90]{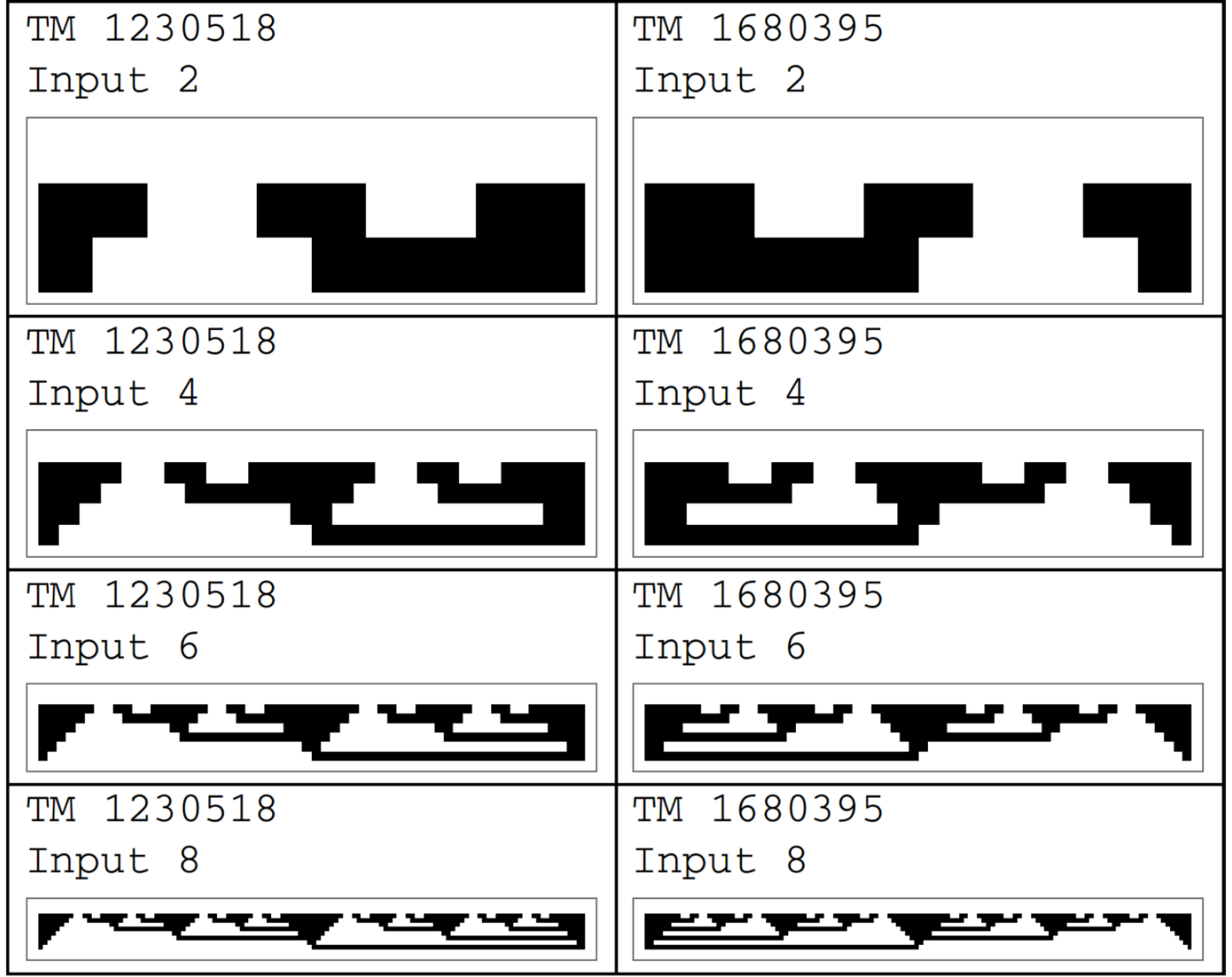}
  \caption{Symmetric performers}
  \label{fig:SymmetricPerformers}
\end{figure}

In particular, Figure \ref{fig:SymmetricPerformers} shows a very curious phenomenon that we call \emph{symmetric performers}. There turns out to be a pair of different TMs so that the space-time diagram on every even input of the one machine is the exact symmetric image of the space-time diagram of the other TM on the same input.

Of course, this can only happen in case the TM computes the tape-identity since the input must equal the output in order to yield a symmetric image. At first, one might be tempted to think that this phenomenon is bound to occur since we can define for each TM $\tau$ its reversed machine $\hat \tau$: replace each instruction $\langle {\sf color}, {\sf state}\rangle \mapsto \langle {\sf color}', {\sf state}', {\sf direction}\rangle$ by its canonical reversal  
\[
\langle {\sf color}', {\sf state}'\rangle \mapsto \langle {\sf color}, {\sf state}, \overline{{\sf direction}}\rangle
\] 
where $\overline{{\sf direction}}$ changes ${\sf right}$ to $\sf left$ and vice-versa. However, note that both machines start in State 1 so that this imposes already a strong condition on possible solutions of symmetric performers.

Let us denote by $\tau$ and $\tilde \tau$ a pair of symmetric performers. It is clear that if a TM $\tau$ terminates on input $x$ it does so in an even number of steps: for each computation where the head moves one to the left (the end of the tape is on the right by our convention), there must be a step where the machine moves one step to the right. In particular, for symmetric performers that terminate in $2n$ many steps on input $x$, we have that if the tape configuration at step $m$ differs from the tape configuration at step $m+1$ in $\tau (x)$, then the head position in step $m$ on $\tau (x)$ is the same as the head position in step $2n-(m+1)$ on $\tilde \tau (x)$.

Indeed it comes as a surprise that all these constraints can be met in (3,2) space, if only just for the even inputs. 
\newpage
\section*{{\LARGE Part III: A brief literature survey}}
\addcontentsline{toc}{section}{Part III}

In the third and final part of the paper, we will try to locate our results within the landscape of known theoretical results that link fractal dimensions to other notions of complexity.

\section{Relations between fractal dimensions and other notions of complexity: An incomplete survey of the literature}\label{section:LiteratureSurvey}

In this paper we have worked with a variant of box-counting dimension and with space and time complexity for processes implemented on Turing machines. These are just some out of a myriad of different complexity measures in the literature. Since eventually the notion of being complex or not is relative to a framework and the ultimate framework in which all these complexity notions can be embedded in is our own cognitive system, on philosophical grounds one can expect relations between the various \emph{a priori} unrelated complexity notions (see \cite{Joosten:2013:NecessityOfComplexity, Joosten:2013:ComplexityFitsTheFittest}). And, indeed, in the literature we find various relations between different notions of complexity.

In this final section we wish to place our results in the context of other results in the literature that link different complexity notions. Our point of departure will be  fractal dimensions and possible relations to complexity notions of a computational nature.

The current section is neither self-contained nor do we pretend to give an exhaustive overview of the literature. Rather, we will try to provide sufficient pointers so that this section at least can serve as a point of departure for a more exhaustive and self-contained study.

\subsection{Box-counting dimension within the landscape of topological and fractal dimensions}\label{section:Dimensions}

In this paper we decided to work with a variant of box-counting dimension since this has many desirable computational properties and applications. Let us first see where box-counting dimension fits into the landscape of various versions of fractal and other dimensions.

Edgar divides geometrical dimensions in two main groups, \emph{topological} and \emph{fractal} dimensions (see \cite{Edgar:1990:MeasureTopologyFractalGeometry}). Topological dimensions are invariants of topological spaces in that they are invariant under homeomorphisms. Moreover, topological dimensions have integer values although some versions allow transfinite (ordinal) values too.

The most basic of all topological dimensions is the so-called \emph{cover dimension} also called \emph{Lebesgue dimension}. In order to describe this dimension we need some additional notions.
  
The \emph{order} of a family $\mathcal A$ of sets is $\leq n$ by definition when any $n+2$ of the sets have empty intersection. We denote this by $o (\mathcal A) \leq n$. We say that $o (\mathcal A)=n$ when $o (\mathcal A) \leq n$ but not $o (\mathcal A) \leq n-1$. Thus, for example, if any two sets in $\mathcal A$ have empty intersection the order of $\mathcal A$ is 0.

The \emph{cover dimension} of a set $S$ is $n$ --we write ${\sf Cov(S) = n}$-- whenever each open covering of $S$ has a refinement of order $n$. Thus, for example, a collection of two separate points in $\mathbb R^n$ has cover dimension 0 since we can separate the points by two disjoint opens. Likewise, any line-like space admits an open cover of order 1, that is, any intersection of three different opens is empty. Similarly we can cover a planar set by open tiles where each row of tiles is shifted to the right, say, w.r.t. the adjacent rows of tiles. This collection of tiles has order two since any collection of four different of such open tiles will be empty.
\medskip

Fractal dimensions on the other hand can have non-integer values. In a sense, the fractal dimension of some object $S$ is an indication of how close $S$ is to some integer-valued dimensional space. Dimension in integer-valued dimensional spaces in a sense express degrees of freedom and as such this provides us an information theoretical focus on dimension. More common is the geometrical focus on (fractal) dimension as for example expressed by Falconer (\cite{Falconer:2003:FractalGeometry}): ``Roughly, dimension indicates how much space a set occupies near to each of its points."

The most fundamental, and most common notion of fractal dimension is that of \emph{Hausdorff dimension} (\cite{hausdorff}) which was introduced already in 1919 building forth upon ideas of Carath\'eodory from 1914 (\cite{Caratheodory:1914}).
    
In order to relate our box-counting dimension to the more common Hausdorff dimension we will outline the definition and some basic properties of Hausdorff dimension.
	  
For $S$ a subset of some metric space we can consider countable open coverings $\mathcal A$ of $S$ and define
  \[
  {\mathcal H}^s_{\varepsilon} (S) : = \inf \sum_{A\in \mathcal A} ({\sf diam}\ A)^s.
  \]
Here, ${\sf diam}\ A$ denotes to the usual diameter of $A$ as the supremum of distances between any two points in $A$. The infimum is taken over all $\mathcal A$ that are countable open $\varepsilon$-covers of $S$. This means that the diameters of the open sets in our cover do not exceed $\varepsilon$. It is essential that we may take the diameters of the open sets in our cover to vary and in particular we can choose them as small as convenient. Next, we define
\[
{\mathcal H}^s (S) : =\lim_{{\varepsilon}\to 0}   {\mathcal H}^s_{\varepsilon} (S).
\]
The main theorem about these $  {\mathcal H}^s (S)$ is that there is a unique $s$ so that 
  
  				\begin{itemize}
				\item ${\mathcal H}^t_{\varepsilon} (S) = \infty$ for $t<s$;
				\item ${\mathcal H}^t_{\varepsilon} (S) = 0$ for $t>s$.
				\end{itemize}
				
This unique $s$ is called the Hausdorff dimension of $S$: ${\sf dim_H}(F)$.
As mentioned, this dimension was introduced in 1919 by Hausdorff (\cite{hausdorff}) and the main theory was later developed mainly by Besicovitch and his students \cite{Besicovitch:1929,BesocovitchPart2,BesocovitchPart3,besicovitchPartIV, besicovitchPartV} so that (\cite{hausdorff}) Mandelbrot often speaks of \emph{Hausdorff-Besicovitch dimension}.

The Hausdorff dimension comes with a natural dual dimension called \emph{packing dimension}. Although the notion of packing dimension is natural and related to the Hausdorff dimension it was only introduced about sixty years later by Tricot (\cite{Tricot:1982}) and Sullivan (\cite{Sullivan:1984}).

The main idea behind packing dimension of some spatio-temporal object $F$ is to somehow measure the volume of disjoint balls one can find so that the center of these balls lie within $F$. As with the case of Hausdorff dimension one parametrizes this concept with the target dimension $s$:
  \[
  {\mathcal P}^s_\delta(F):=
 \{ \sup \sum_{i} |B_i| \mid \{ B_i\}_i \mbox{ are disjoint balls at radii $\leq \delta$ and center in $F$}\}
  \]
Since $\lim_{\delta \to 0} {\mathcal P}^s_\delta(F)$ is not a measure (this is easy to see by considering countable dense sets of some $F$ with positive dimension)
 one applies a standard trick which transforms this into a measure by defining
  \[
  {\mathcal P}^s (F) := \inf_{\{F_i\}_i} \{  \sum_i \lim_{\delta \to 0} {\mathcal P}^s_\delta(F_i) \mid F \subseteq \bigcup_{i=1}^\infty F_i\}. 
  \]
Here the infimum is taken over countable collections of sets $F_i$  so that 
$F \subseteq \bigcup_{i=1}^\infty F_i$. The main theorem of this notion ${\mathcal P}^s (F)$ shows that Packing dimension is in a sense dual to Hausdorff dimension: There is a unique $s$ so that 
  				\begin{itemize}
\item  ${\mathcal P}^t (F) = 0$ for $t<s$;
\item $ {\mathcal P}^t (F)= \infty$ for $t>s$.
				\end{itemize}
				
This unique $s$ is called the packing dimension of $F$ and we write ${\sf dim_P}(F)$. It is not hard to see that Packing dimension is an upper bound to Hausdorff dimension, that is, ${\sf dim_H}(F) \leq {\sf dim_P}(F)$.
\medskip
	
A fundamental property that is not hard to prove of the dimensions we have seen so far is that: ${\sf Cov}(F) \leq {\sf dim_H}(F)$. Mandelbrot defines a \emph{fractal} to be any set $F$ with ${\sf Cov}(F) < {\sf dim_H}(F)$. However, this notion of fractal is often considered (also by Mandelbrot himself) a notion of fractal that is too broad, since it admits ``true geometric chaos". J. Taylor proposes (see \cite{Taylor:1986}) to denote by fractals only Borel sets $F$ for which	${\sf dim_H}(F) = {\sf dim_P}(F)$. 
\medskip

We can now see how box-counting dimensions (or box dimensions for short) naturally fit the scheme of fractal dimensions we have seen above. In particular, the box dimension is like Hausdorff dimension only that we now cover by balls/boxes of \emph{fixed size} rather than by ball of flexible size not exceeding some maximum value $\varepsilon$. 
    
 Alternatively and equivalently, in order to define the box dimension, we can divide space into a regular mesh with mesh-size $\delta$ and count how many cells $N_\delta(F)$ are hit by a set $F$. Then, we define ${\mathcal B}^s_\delta (F) : = N_\delta (F) \delta^s$ and ${\mathcal B}^s(F) : = \liminf_{\delta \to 0} N_\delta (F) \delta^s$.
  
 Again, there is a cut-off value $s_0$ so that ${\mathcal B}^s(F) = \infty$ for $s<s_0$ and ${\mathcal B}^s(F) = 0$ for $s>s_0$.
  This cut-off value is given by 
  \[
  \liminf_{\delta\to 0} \frac{\log (N_\delta (F))}{\log(1/\delta)}.
  \]
which is close to the notion we started out with in this paper in Definition \ref{definition:BoxDimension}. Inspired by this cut-off value, we define 
\[
{\sf \underline{dim}_B}: =   \displaystyle \liminf_{\delta\to 0} 
  \frac{\log (N_\delta (F))}{ \log(1/\delta)}
  \] and 
  \[
  {\sf \overline{dim}_B}: =   
  \displaystyle \limsup_{\delta\to 0} \frac{\log (N_\delta (F))}{\log(1/\delta)}.
  \]

In case ${\sf \underline{dim}_B}(F) = {\sf \overline{dim}_B}(F)$ we call this the box-counting dimension: ${\sf {dim}_B} (F)$ which now exactly coincides with Definition \ref{definition:BoxDimension}. One can easily show that box dimension always provides an upper bound to Hausdorff dimension. Moreover, box dimension has many desirable computational properties thereby being amenable for computer applications.
  
Notwithstanding the good computational behavior, box dimension has various undesirable mathematical properties:  in particular, a countable union of measure zero sets can have positive box dimension. For example, one can show that in $\mathbb R$ with the standard topology we have ${\sf dim_B} \{ 0, \frac{1}{2}, \frac{1}{3}, \frac{1}{4}, \ldots \} = \frac{1}{2}$ which is of course highly undesirable.
  
Mathematically, this undesirable properties can be impaired with the same trick that was applied to the packing dimension by defining \emph{modified box dimension} as
\[
  {\sf \underline{dim}_{MB}}(F) := \inf_{\{ F_i\}} \{ \sup_i {\sf \underline{dim}_B (F_i)} \mid F \subseteq \bigcup_{i=1}^\infty F_i \} \ \ \ \mbox{ and }
\]
  
  \[
  {\sf \overline{dim}_{MB}}(F) := \inf_{\{ F_i\}} \{ \sup_i {\sf \overline{dim}_B (F_i)} \mid F \subseteq \bigcup_{i=1}^\infty F_i \}
  \]

But, of course, by doing so, we would loose all the good computational properties. In general, we have that 
\[
{\sf dim_H}(F)\leq {\sf \underline{dim}_{MB}}(F)\leq {\sf \overline{dim}_{MB}}(F) = {\sf dim_P}(F) \leq {\sf \overline{dim}_B}(F)
\]
and it is known that none of the inequalities can be replaced by equalities.
However, we note that under Taylor's definition of fractal, the first four dimensions collapse and modified box dimension is an equivalent of Hausdorff dimension and indeed the modified box-counting dimension is a natural quantity to consider.

Moreover, if $F$ has a lot of self-similarity, then modified box-counting dimension is actually equal to the plane box counting dimension:

\begin{proposition}
Let $F\subseteq \mathbb{R}$ be compact so that for any open set $V$ we have ${\sf \overline{dim}_B} (F) = {\sf \overline{dim}_B} (F\cap V)$, then ${\sf \overline{dim}_B} (F) = {\sf \overline{dim}_{MB}} (F)$.
\end{proposition}

So in various situations, box counting coincides with Hausdorff dimension. The most famous example is probably that this equality holds for the Mandelbrot set. In addition, there are various other situations where box-counting and Hausdorff dimension coincide (\cite{Staiger:1998,Staiger:2005}).

\subsection{Computability properties of fractals}\label{section:ComputabilityPropertiesJuliaSets}

As a first link between fractals and computability properties we want to mention that of various fractal objects one has studied the computational complexity.

Probably the most famous examples of fractals are Julia sets and the corresponding `roadmap Mandelbrot set'.
  Let us briefly recall some basic definitions.
  By ${\sf FJ}(f)$ we denote the \emph{filled Julia set} of a function $f$ defined on the complex numbers. This set ${\sf FJ}(f)$ is defined as the set of values $z$ in the domain of $f$ on which iterating $f$ on $z$ does not diverge. That is,
  \[
{\sf FJ}(f) := \{ z \mid  \limsup_{n\to \infty} |f^n(z)|< \infty \}.
  \]
By $J(f)$ --the \emph{Julia set of $f$}-- we denote the boundary of ${\sf FJ}(f)$.
Following C.T. Chong (\cite{ChongSlides:2009}), we can consider $f_\theta (z) = z^2 + \lambda z$ with $\lambda = e^{2\pi i \theta}$ and $\theta \notin \mathbb Q$. Using this notation, the corresponding Julia sets are denoted by $J_\theta$.

One can express that $J_\theta$ is well-behaved by saying that it has a \emph{Siegel disk} at $z=0$. Basically, this says that $f$ is locally linearizable at $z=0$ by a rotation and we refer the reader to e.g.\ \cite{Milnor:2006} for further details.

The \emph{Turing degree} of a set is an indication of how complicated that set is. A set of natural numbers $A$ is of Turing degree at most that of $B$ --we write $A\leq_TB$-- if the question about $x\in A$ can be decided on an idealized computer using various queries of the form $y\in B$. We say that two sets $A$ and $B$ have the same Turing degree --we write $A\equiv_TB$-- whenever both $A\leq_T B$ and $B\leq_T A$.

Likewise, we say that a set $B$ is \emph{computably enumerable} --or c.e.\ for short-- in $A$ if we can, using an idealized computer, enumerate all the elements of $B$ using queries about $A$. Note that enumerability of $B$ does not give a procedure to decide membership. It only guarantees you that if some element belongs to $B$, then at some stage it will be enumerated in the enumeration.

We call a set $A$ \emph{recursive}, \emph{computable} or simply \emph{decidable} if we can decide with an idealized computer without any oracles whether $x\in A$ or not for any $x\in \mathbb N$. Likewise, we call a set $A$ simply c.e.\ when it is c.e.\ in the empty set $\varnothing$.

The Turing degree of a set --the equivalence class under $\equiv_T$ so to say-- is a robust notion in various ways. For example, it makes sense to speak of `being c.e.\ in the degree of $A$' whence we will often refrain from distinguishing $A$ from its corresponding degree.

We can conceive a real number as a set of natural numbers. Let us restrict ourselves to the real interval $[0,1]$. Then we can conceive any real number $a$ in this interval as a set by looking at the binary expansion of $a$ and using this string $0,a_0a_1a_2\ldots$ to define a set $A$ where $i \in A$ iff $a_i =1$.
So by this identification it makes sense to speak of the Turing degree of a real number.

We will shortly discuss that one can set up real analysis in such a way that it also makes sense to speak about the Turing degree of non-discrete objects like $J_\theta$. Braverman and Yampolsky follow in \cite{BravermanYampolsky:2009:ComputabilityJuliaSets} an approach of what is called \emph{Constructive Analysis} as initiated by Banach and Mazur (\cite{BanachMazur:1937}), with influences of Markov (\cite{Markov:1962}). The main idea behind this constructive analysis is that we can conceive continuous objects as entities that we can computably approximate to the precision that we require (see e.g.\ \cite{weihrauch2000computable} for an overview).

Braverman and Yampolsky have studied (see \cite{BravermanYampolsky:2009:ComputabilityJuliaSets}) the relations between the Turing degree of $\theta$ and that of $J_\theta$. In particular they prove that $\mathbf b$ is a c.e.\ Turing degree if and only if it is the degree of $J_\theta$ with $\theta$ recursive so that $J_\theta$ has a Siegel disk.

C.T. Chong has generalized this result (\cite{ChongSlides:2009}): Let $\mathbf c$ be a Turing degree. For every ${\mathbf d} \geq {\mathbf c}$ we have that $\mathbf d$ is c.e.\ in $\mathbf c$ if and only if it is the degree of a Julia set $J_\theta$ with Siegel disk and ${\sf deg}(\theta)= \mathbf c$.
  
It is good to stress that all these results are sensitive to the underlying model of computation and real analysis and the results would change drastically if one were to switch to other models like the so-called Blum-Schub-Smale model (see \cite{BlumCuckerShubSmale:1998}).

The results presented in this subsection relate the Turing complexity of the fractal to the complexity of the parameter generating it. However there are no links from the Turing degrees of the Julia sets to the corresponding dimensions. In the next subsection we will discuss various results of this sort.

\subsection{Effective dimension and computations}\label{section:EffectiveDimensionAndComputations}

In this subsection we will present certain results that relate Hausdorff dimension to other notions of complexity. In order to do so we will first rephrase the notion of Hausdorff dimension in the setting of binary strings. Next, we shall shall define so-called effectivizations of Hausdorff dimension. It is these effectivizations which can be related to other notions of complexity. Again, this subsection will be far from self-contained. We refer the reader to \cite{DowneyHirschfeldt:2010} for further details. And actually the presentation here is largely based on this treatise (mainly Chapter 13) and we shall closely follow it in structure.
	  
Thus, let us reformulate the definition of Hausdorff dimension in the realm of binary sequences, i.e., in the realm of Cantor space which we shall denote by $2^\omega$. We shall interchangeably speak of sequences or of reals when we refer to elements of Cantor space. The collection of finite binary strings we shall denote by $2^{<\omega}$.

For $\sigma \in 2^{<\omega}$ we denote the length of $\sigma$ as $| \sigma |$. For $\sigma \in 2^{<\omega}$ we define $\lb \sigma \rb := \{ \sigma \tau \mid \tau \in 2^\omega \}$ where $\sigma \tau$ denotes just string concatenation. 
Whenever we shall consider Cantor space as a topological space, we shall consider the topology generated by the basic open sets of the form $\lb \sigma \rb$.
For $\Sigma \subseteq 2^{<\omega}$ we define $\lb \Sigma \rb := \bigcup_{\sigma\in \Sigma} \lb \sigma \rb$.

Thus, for any $R \subseteq 2^\omega$ we define an \emph{$n$-cover} of $R$ to be a set $\Sigma \subseteq 2^{\geq n}$ so that $R\subseteq \lb \Sigma \rb$. Cantor space can be endowed with a measure in the standard way by defining $\mu (\lb \sigma \rb) = 2^{-|\sigma|}$. Thus, in analogy to Section \ref{section:Dimensions} we now define:
\[
{\mathcal H}^s_n (R) := \inf \{ \sum_{\sigma \in \Sigma} 2^{-s |\sigma|} \mid \Sigma \mbox{ an $n$-cover of $R$}\}
\]
and ${\mathcal H}^s (R) := \displaystyle \lim_{n\to \infty} {\mathcal H}^s_n (R)$.
So, as before, we define ${\sf dim_H} (R) := \inf \{ s \mid {\mathcal H}^s (R) = 0 \}$. It is easy to see that for every $r\in [0,1]$ there is $R\subseteq 2^\omega$ with ${\sf dim_H}(R) = r$.

Within the context of Cantor space we shall now give a definition of what is called \emph{Effective Hausdorff dimension}. The effective pendant is defined via
\[
{\mathcal {EH}}^s_n (R) := \inf \{ \sum_{\sigma \in \Sigma} 2^{-s |\sigma|} \mid \Sigma \mbox{ a \emph{c.e.} $n$-cover of $R$}\}
\]
and ${\mathcal {EH}}^s (R) \ := \ \displaystyle \lim_{n\to \infty} {\mathcal EH}^s_n (R)$, so that the \emph{effective Hausdorff dimension} is defined as ${\sf dim_{EH}} (R) := \inf \{ s \mid {\mathcal {EH}}^s (R) = 0 \}$.
 
One can now show (\cite{Lutz:2003}) that for every computable real $r \in [0,1]$, there is a set $R\subseteq 2^\omega$ with ${\sf dim_{EH}}(R) = r$. By a theorem of Hitchcock's we have that for important subsets $F$ of Cantor space it holds that ${\sf dim_{H} (F)}= {\sf dim_{EH} (F)}$:
 
\begin{theorem}[Hitchcock \cite{Hitchcock:2005}] Let $F$ be a countable union of $\Pi^0_1$ classes\footnote{A subset $A$ of Cantor space is a $\Pi^0_1$ class if it is the collection of paths for some computable tree. An alternative definition requires that for some computable relation $R$ we have $A:= \{ \sigma \in 2^\omega \mid \forall n\ R(\sigma {\upharpoonright} n) \}$.} of Cantor space, then ${\sf dim_{H} (F)}= {\sf dim_{EH} (F)}$. 
\end{theorem}

In the same paper Hitchcock also proves an equality for $\Sigma^0_2$ classes and \emph{computable} \emph{ Hausdorff dimension} (covers are required to be \emph{computable} rather than c.e.). So, for some objects Hausdorff dimension and effective Hausdorff dimension coincide.

However, for other important classes they differ. In particular we have that the Hausdorff dimension of any sequence in Cantor space equals zero. However, there may be no simple effective covers around so that a single sequence can have positive effective Hausdorff dimension.

There is a link between Turing degrees and effective Hausdorff dimension albeit this link is not very straight-forward.
Recall that for $A\in 2^\omega$ we have ${\sf dim_H}(A) = 0$ but that we can have ${\sf dim_{EH}}(A) > 0$ when no simple effective covers are around.

Thus, in a sense, having non-zero effective Hausdorff dimension is an indication of containing complexity. And in fact it can be shown that if $A \in  2^\omega$ with ${\sf dim_{EH}}(A)>0$, then $A$ can compute a non-recursive function. To be more precise, $A$ can compute a fix-point free function $f$ (that is, a function $f$ so that $W_{f(e)}\neq W_e$ for all numbers $e$) by results of Terwijn, Jockush, Lerman, Soare and Solovay (\cite{Terwijn:2004,JockushLermanSoareSolovay:1989}).
  
%

This result establishes a relation between effective Hausdorff dimension and computational complexity in the guise of degrees of undecidability. However, the relation between effective dimension and computable content is not monotone nor simple. In particular, one can show that if ${\sf dim_{EH}}(A)=\alpha$, then there exist sets $B$ of arbitrary high Turing degree with  ${\sf dim_{EH}}(B)=\alpha$. However locally, Hausdorff dimension can provide an upper bound to Turing degrees:
  
\begin{theorem}[Miller \cite{miller2011extracting}]\label{theorem:Miller} 
Let $r$ be a left-c.e.\ real\footnote{We omit the technical details here and refer to \cite{DowneyHirschfeldt:2010} for them. However, one can think of a left-c.e.\ real as a c.e.\ real that does converge but for which we cannot computably estimate the rate of convergence.}. There is a $\Delta^0_2$-definable set $R\in  2^\omega$ with ${\sf dim_{EH}}(R) = r$ so that moreover 
\[
  A \leq_T  R \ \ \Rightarrow \ \ {\sf dim_{EH}}(A) \leq r.
\] 
\end{theorem}

It is exactly this kind of results that we are interested here in this section: theorems that relate different notions of complexity. Another classical result links Kolmogorov complexity to effective Hausdorff dimension. Let us briefly and loosely define Kolmogorov complexity referring to e.g.\ \cite{LiVitanyi:book} for further details.
	  
For a string $s\in 2^{<\omega}$ the Kolmogorov complexity $K(s)$ is roughly the length of the shortest program that outputs $s$ when computed on a particular universal Turing machine. Of course this is dependent on a particular choice of a universal Turing machine, but different choices of a universal Turing machine only manifest itself in an additive constant in $K$. The relation between Kolmogorov complexity and effective Hausdorff dimension is given by a theorem of Mayordomo:
	
\begin{theorem}[Mayordomo \cite{Mayordomo:2002}]\label{theorem:Mayordomo}
Let $A$ be a sequence in	Cantor space and let $A{\upharpoonright} n$ denote the first $n$ bits of this sequence. 
\[
{\sf dim_{EH}}(A) = \liminf_{n\to \infty} \frac{K(A{\upharpoonright} n)}{n}.
\]  
\end{theorem}

Moreover, there is a link from effective Hausdorff dimension to a notion that is central to probability theory: \emph{Martingales}. Martingales indicate expected outcomes of betting strategies. Lutz introduced in \cite{Lutz:2000} an adaptation of this notion that can be linked to effective Hausdorff dimension (or constitute an alternative definition for that matter).

\begin{definition}
An $s$-gale is a function $d: 2^{<\omega} \to \mathbb R^{\geq 0}$ such that $d(\sigma) = \frac{d(\sigma 0) + d(\sigma 1)}{2^s}$.
\end{definition}
  
This is a generalization of `gales' (as introduced/simplified by \cite{Levy:1937}) where $d(\sigma) = \frac{d(\sigma 0) + d(\sigma 1)}{2}$ expresses a certain fairness condition of the betting strategy. In particular, one can see $d$ as a pay-off function where the equality expresses that your expectation is to not loose nor gain money.
  
We say that a certain gale $d$ \emph{succeeds} on $A$ whenever $\limsup_{n\to \infty} d (A\upharpoonright n) = \infty$.
  
The \emph{Success set} of $d$ is the collection of all $A$ on which $d$ succeeds and is denoted by $S[d]$. The link from Hausdorff dimension to these gales is given by a theorem by Lutz:
  
\begin{theorem}[Lutz \cite{Lutz:2000}]\label{theorem:Lutz}
  \[
  {\sf dim_{EH}}(X) = \inf \{ q \in \mathbb Q \mid X \subseteq S[d] \mbox{ for some $q$-gale $d$}\}.
  \]
\end{theorem}  

In the context of this paper it is good to mention that other notions of dimension also have their effective counterparts. In particular, Reimann studied an effectivization of box counting dimension in \cite{Reimann:2004} and the corresponding relations to the other complexity notions are similar to the ones mentioned here.

Also, for various dimensions the \emph{computable} versions have been studied, where the open covers are no longer required to be c.e.\ but rather computable. We refer the reader to \cite{DowneyHirschfeldt:2010} for further details.

\subsection{Our results}

In this section we have tried to present a selection of readily accessible results in the literature that relate fractal dimension to other notions of complexity.

We mentioned results of Braverman and Yampolsky and the generalization thereof by Chong in \ref{section:ComputabilityPropertiesJuliaSets}. These results related computational (Turing degrees) properties of Julia sets to the same computational properties (Turing degrees) of the parameter that generates the Julia set. In a sense, this is not a result that links Hausdorff dimension to a different notion of complexity. However, since it is one of the few results on computational properties of fractals, we have decided to include it in the overview.

Next, in \ref{section:EffectiveDimensionAndComputations} we only worked in the realm of Cantor space. There we considered an effectivization of Hausdorff dimension and this notion was linked to Turing degrees as in Theorem \ref{theorem:Miller}, to Kolmogorov complexity as in Theorem \ref{theorem:Mayordomo} and to martingales as in Theorem \ref{theorem:Lutz}.

Effective Hausdorff dimension however works with highly idealized notions relatively high up in the computational hierarchy. Our results involve complexity classes which are more down-to-earth like $\sf PTIME$ and $\sf PSPACE$. Moreover, in our theorems the two complexity notions that are related --computational complexity versus fractal dimension-- are not applied to exactly the same object as is the case in the earlier mentioned results. Rather they link two attributes of a small Turing machine: the computational (runtime) complexity on the one hand and the fractal dimension of the corresponding space-time diagrams on the other hand. It is in these respects that, to the best of our knowledge, our results are new in their kind.

\subsection*{Acknowledgements}
JJ wishes to thank Jos\'e Mar\'{\i}a Amig\'o Garc\'{\i}a for his kind support and an invited research stay at the University of Miguel Hernandez in Elche. JJ received support from the Generalitat de Catalunya under grant number 2009SGR-1433 and from the Spanish Ministry of Science and Education under grant numbers MTM2011-26840, and MTM2011- 25747. We acknowledge support from the project FFI2011-15945-E (Ministerio de Econom\'{\i}a y Competitividad, Spain).


\bibliographystyle{plain}
\bibliography{fractal}

\end{document}